\documentclass[11pt,a4paper]{article}

\usepackage{amsmath,amssymb,amsthm}
\usepackage{graphicx}
\usepackage{caption}
\usepackage{cite}
\usepackage{hyperref}
\usepackage{url}
\usepackage{geometry}
\newtheorem{theorem}{Theorem}[section]
\newtheorem{lemma}[theorem]{Lemma}
\newtheorem{proposition}[theorem]{Proposition}

\newtheorem{remark}[theorem]{Remark}

\newtheorem{assumption}{Assumption}[section]

\title{Photon surface extensions in general spherical dust collapse}

\author{Roberto Giamb\`o$^{1,2}$\thanks{E-mail: \texttt{roberto.giambo@unicam.it}}\,\, and Camilla Lucamarini$^{1}$\\
\small $^1$ School of Science and Technology, Mathematics Division, Universit\`a di Camerino,\\
\small Via Madonna delle Carceri 8, Camerino, 62032, MC, Italy.\\
\small $^2$ Osservatorio Astronomico di Brera, INAF, Via Brera 28, Milano, 20121, Italy.}

\date{}

\begin{document}

\maketitle

\begin{abstract}
We extend the analysis of photon surfaces in spherical dust collapse to the general,
non-marginally bound case, i.e.\ allowing the \textit{energy function} $k(x)$ of the
Lema\^{\i}tre--Tolman--Bondi (LTB) model to be non-zero. Starting from the
dynamical-systems formulation developed in our previous work for the marginally
bound case~\cite{GiamboLucamarini2024}, we derive the photon surface equations
for the LTB metric with $k(x)\neq 0$, and we recast the geometric condition for a
timelike hypersurface to be a photon surface as a non-autonomous first-order
dynamical system. Even though the LTB evolution equation does not integrate in
closed form when $k(x)\neq 0$, the implicit solutions available in the
literature, together with comparison theorems for ordinary differential
equations, are sufficient to reduce the problem to a single scalar equation,
governed by one quantity which vanishes identically for homogeneous data. Using
this reduction we show that the extension of the exterior photon sphere $r=3M$
into the collapsing cloud is unique, and null, if and only if the central
singularity is covered. When the singularity is naked, uniqueness fails: besides
the null hypersurface there exists a one-parameter family of genuinely timelike
photon surfaces matching $r=3M$, all of them reaching the central singularity. 
The threshold separating the two behaviours coincides with the one separating
black hole and naked singularity end states. We also show that the null photon surface reaches the central singularity if and only if the singularity is naked. The structural dichotomy between the naked and covered end states is discussed in
connection with possible observational signatures in the early-time evolution of
black-hole shadows.
\end{abstract}


\section{Introduction}
\label{sec:intro}

The notion of photon sphere classically arises in static and spherically symmetric
black-hole spacetimes, such as the Schwarzschild solution, as a particular
timelike hypersurface on which null geodesics initially tangent to the surface
remain tangent throughout their evolution. In the Schwarzschild geometry the
photon sphere coincides with the timelike cylinder $r=3M$, and it plays a
fundamental role in the analysis of the shadow of a black hole as observed by a
distant observer~\cite{EHT2019,EHT2024,VirbhadraEllis2002}.

A natural geometric generalisation of this notion to spacetimes lacking specific
symmetries or Killing structures is the concept of \emph{photon surface},
introduced by Claudel, Virbhadra and Ellis~\cite{Claudel2001}. By their
definition, a photon surface is a timelike, embedded hypersurface such that every
null geodesic initially tangent to it remains tangent. The construction is
intrinsic, makes no reference to the existence of Killing vectors, and is
therefore well suited to fully dynamical settings such as gravitational collapse,
where no global timelike Killing field is available~\cite{Cao2021}.

The detection of black holes through their shadows by the Event Horizon
Telescope~\cite{EHT2019,EHT2024} has further motivated the systematic study of
photon surfaces in dynamical spacetimes, since the photon surface is the natural
geometric object underlying the formation of the shadow in non-stationary
backgrounds. One physically rich setting in which photon surfaces can be
investigated dynamically is gravitational collapse, which provides at the same
time a concrete model of black-hole formation and an arena in which the cosmic
censorship conjecture~\cite{Penrose1965} can be tested.

The LTB dust collapse model in the \emph{marginally bound} case, i.e.\ with
vanishing energy function $k(x)$, was studied from this point of view in
\cite{GiamboLucamarini2024} (hereafter \textsc{Paper~I}). There it was shown that
the condition for a timelike $SO(3)$-invariant hypersurface in a general
spherically symmetric spacetime to be a photon surface can be reformulated as a
non-autonomous first-order dynamical system, with a single auxiliary variable
$y(t)$ encoding the tilt of the hypersurface relative to the comoving observer. 
Within that framework it was shown that the extension of the exterior photon
sphere $r=3M$ into the cloud is unique -- and then null, generated by outgoing
radial null geodesics -- if and only if the central singularity is covered, and
that the null extension reaches the central singularity precisely when the
singularity is naked.

The assumption $k\equiv 0$ was adopted in \textsc{Paper~I} mainly because it
allows the LTB evolution equation to be integrated in closed elementary form,
yielding the explicit expression $r(t,x)=x\bigl(1-\mu(x)t\bigr)^{2/3}$ for the
areal radius. This explicit solution greatly simplifies the quantitative
analysis, both for the dynamical system governing the photon surface and for the
comparison arguments needed to study its behaviour near the central singularity.
From a physical and mathematical standpoint, however, it is natural to ask
whether the qualitative picture established in \textsc{Paper~I} survives when
the marginally bound assumption is dropped. This is a non-trivial question: in the non-marginally bound case the LTB equation
does not admit a closed-form elementary solution, but only either a parametric
(cycloidal or hyperbolic) or an implicit one, and the photon surface equations
acquire additional terms depending on $k(x)$ and on its derivative $k'(x)$.

The aim of the present paper is twofold. On the one hand, we show that the
dynamical-systems framework adapts naturally to the general LTB model, and that
several of the qualitative conclusions of \textsc{Paper~I} are the same: the
extension of $r=3M$ into the cloud is unique, and then null, if and only if the
central singularity is covered; the null extension reaches the singularity
precisely when the latter is naked; and in that case an infinite family of
outgoing null geodesics emanates from the singularity and escapes to the
exterior, so that the photon surface fails to cover it. That these statements
survive is not immediate, since the photon surface equations do change -- they
acquire terms in $k$ and $k'$, and the closed-form solution on which the
marginally bound analysis rests is no longer available. The absence of a closed-form solution forces a different
method: the photon surface system reduces to a single scalar
equation whose entire content is carried by the quantity
\[
  D=\frac{r_x\,r_t}{r}-r_{tx},
\]
for which we obtain an exact expression, and which vanishes identically precisely
when the rescaled data $m(x)/x^3$ and $k(x)/x^2$ are constant. The obstruction to
uniqueness is therefore a pure inhomogeneity effect, of the pair of free
functions and not of the density alone. But uniqueness of the null extension turns
out to be equivalent to the central singularity being covered, so that
uniqueness itself becomes a characterisation of the black-hole case, and the
threshold is the minimum of an explicit function, $g(\chi)=(2+\chi)/[\chi^3(1-\chi)]$. 
For instance, 
a homogeneous mass profile $m(x)=m_0x^3$ gives, for $k\equiv0$, the
Oppenheimer--Snyder solution: a black hole, with a unique photon surface. But the energy
function alone, at fixed and even homogeneous density profile, can  decide
between uniqueness or not -- a phenomenon that no marginally bound
treatment can exhibit.

The paper is organized as follows. Section~\ref{sec:dustgeneral} reviews the
general LTB model with $k(x)\neq 0$, fixes the standing assumptions on the data,
and collects the facts about the implicit solution and the causal structure of the
central singularity that will be needed in the sequel.
Section~\ref{sec:photonsurfaces} specializes the photon surface equations to this
setting, reduces the system to the single scalar equation \eqref{eq:ustar},
proves the key Lemma~\ref{lem:main} in the covered case, and establishes the dichotomy of Theorem~\ref{thm:nullextension}. Section~\ref{sec:analysis} reviews
the resulting null geodesic ODE separately in the covered and naked singularity
cases, by means of sub- and supersolution comparison arguments, and describes the
family of timelike photon surfaces arising in the naked case.
Section~\ref{sec:conclusions} summarizes the results and discusses their physical
implications, in particular suggesting some possible differences between the
early-time shadows of black holes and naked singularities formed in collapse.

\section{Dust collapsing spacetimes: the general case}
\label{sec:dustgeneral}

\subsection{The LTB metric and its evolution equation}

We consider a spherically symmetric spacetime whose interior is filled by a
pressureless perfect fluid (dust), with energy--momentum tensor
$T = -\rho\, dt\otimes \partial_t$, where $\rho=\rho(t,x)$ is the rest-mass
density~\cite{Joshi2012,Singh1996}. The Einstein field equations with this
matter content reduce, in synchronous comoving coordinates $(t,x,\theta,\varphi)$,
to the Lema\^{\i}tre--Tolman--Bondi (LTB) metric
\begin{equation}
  g = -\,dt^2 + \frac{r_x^2}{1+k(x)}\,dx^2 + r^2\,d\Omega^2,
  \label{eq:LTBmetric}
\end{equation}
where $r=r(t,x)$ is the areal radius, $r_x=\partial r/\partial x$, $k(x)>-1$ is
the so-called \emph{energy function}, and
$d\Omega^2 = d\theta^2 + \sin^2\!\theta\,d\varphi^2$ is the round metric on the
unit two-sphere. The evolution of the areal radius is governed by the single PDE
\begin{equation}
  r_t(t,x) = -\sqrt{k(x) + \frac{2m(x)}{r(t,x)}},
  \label{eq:LTBevol}
\end{equation}
where $m(x)$ is the Misner--Sharp mass function and the minus sign selects the
collapsing branch ($r_t<0$). The functions $k(x)$ and $m(x)$ are freely
prescribable initial data, subject to standard regularity requirements at the
centre and to the physical condition that the energy density $\rho$ be
non-negative.

The sign of $k(x)$ classifies the collapsing shells:
\begin{itemize}
  \item $k(x)=0$: \emph{marginally bound} shells, with zero total mechanical
  energy. This is the case treated in \textsc{Paper~I}~\cite{GiamboLucamarini2024};
  it is included as a limiting case of the present analysis.  
  \item $k(x)<0$: \emph{bound} shells, with negative total energy; the collapse
  is decelerated by the gravitational potential.
  \item $k(x)>0$: \emph{unbound} shells, with positive total energy; the shell
  is asymptotically free at infinity.
\end{itemize}

The interior LTB region is defined for $x\in[0,x_b]$, where $x_b$ will be chosen small enough that the initial data satisfy some appropriate conditions discussed below. This interior region is matched to a
Schwarzschild exterior across the timelike hypersurface
$\Sigma=\{x=x_b\}$. The Darmois junction conditions~\cite{Darmois1927}
reduce, in this setting, to the requirement that the Misner--Sharp mass be
continuous across $\Sigma$, i.e.\ that the Schwarzschild mass coincide with the
boundary value $M=m(x_b)$; this condition is automatically propagated by
equation~\eqref{eq:LTBevol}. We refer to~\cite{GiamboLucamarini2024,Magli1997,
Israel1966,MarsSenovilla1993} for further details on the matching procedure.

\subsection{The marginally bound vs.\ the general case}\label{sec:mbvsgen}

In the marginally bound case $k\equiv 0$, equation~\eqref{eq:LTBevol} integrates
in closed elementary form, yielding
\begin{equation}
  r(t,x) = x\bigl(1 - \mu(x)\,t\bigr)^{2/3}, \qquad
  \mu(x) := \tfrac{3}{2}\sqrt{\tfrac{2m(x)}{x^3}},
  \label{eq:explicitk0}
\end{equation}
where the initial condition is $r(0,x)=x$, and the singularity time of the shell
labeled by $x$ is $t_s(x)=1/\mu(x)$. In particular, the central singularity
forms at $t_s(0)=1/\mu(0)$, which we may normalize to unity without loss of
generality.

For $k(x)\neq 0$, the closed-form expression~\eqref{eq:explicitk0} is no longer
available; nevertheless, 
equation \eqref{eq:LTBevol} can be integrated to give $r(t,x)$ in  the implicit form (see \cite{MenaNolanTavakol2004} or \cite[Section 18.3]{PlebanskiKrasinski2006})
\begin{equation}\label{eq:solk}
t_s(x)-t=\sqrt{\frac{r^3}{2m(x)}}\,\Gamma\left(-\frac{k(x) r}{2m(x)}\right),
\end{equation}
where the function $\Gamma$ is defined, for real $z<1$, by
\begin{equation}\label{eq:3.13}
  \Gamma(z)=\int_0^1\frac{\sqrt{\sigma}}{\sqrt{1-z\,\sigma}}\,\mathrm d\sigma .
\end{equation}
Written in this form, $\Gamma$ is manifestly real, positive and smooth on
$(-\infty,1)$, it is strictly increasing, and no choice of complex branches is
needed. In particular $\Gamma(0)=\tfrac23$, which is the marginally bound value.
For $0<z<1$ the integral can be evaluated in closed form as
\begin{equation}\label{eq:Gammaclosed}
  \Gamma(z)=\frac{\arcsin\sqrt{z}}{z^{3/2}}-\frac{\sqrt{1-z}}{z},
\end{equation}
while for $z<0$ the corresponding real expression is
\begin{equation}\label{eq:Gammaneg}
  \Gamma(z)=-\,\frac{\operatorname{arcsinh}\sqrt{-z}}{(-z)^{3/2}}-\frac{\sqrt{1-z}}{z}.
\end{equation}
Note the sign of the first term in \eqref{eq:Gammaneg}: with this sign one has
$\Gamma(z)\to\tfrac23$ as $z\to0^-$, consistently with \eqref{eq:3.13}.

The bound case $k(x)<0$ corresponds to $z>0$, while the unbound case $k(x)>0$
corresponds to $z<0$; the common limit $z=0$ is the marginally bound case $k=0$.

In the equation \eqref{eq:solk}, the function $t_s(x)$ is an integration function that can be specified, as usual, assuming the initial condition $r=x$ at $t=0$:
\begin{equation}\label{eq:ts}
t_s(x)=\sqrt{\frac{x^3}{2m(x)}}\,\Gamma\left(-\frac{k(x) x}{2m(x)}\right).
\end{equation}
Since the righthand side of \eqref{eq:solk} tends to $0$ as $r\to 0^+$, the function $t_s(x)$
expresses the comoving time when 
the shell labeled by $x$ collapses to $r=0$.  

Let us discuss the conditions to impose on the data $m(x)$ and $k(x)$ in order to have a complete inhomogeneous collapse. First of all, we assume  standard regularity conditions on the initial data: $m(x)/x^3$ smooth and positive, $k(x)/x^2$ smooth. More specifically, we assume the free functions $m(x)$ and $k(x)$ admit Taylor expansions of the form
\begin{align}
  m(x)&=x^3(m_0+m_p x^{p}+o(x^{p})), \qquad m_0>0,\,m_p\in\mathbb R,\ p\in\mathbb{N}\setminus\{0\},
         \label{eq:massprofile}\\[4pt]
  k(x) &= x^2(k_0+k_qx^q + o(x^q)), \qquad k_0,\,k_q\in\mathbb R,\ q\in\mathbb{N}\setminus\{0\},
         \label{eq:kprofile}
\end{align}
where either $m_p\ne0$, or $m(x)=m_0x^3$ identically, in which case we set
$p=+\infty$; and likewise for $k$, with the further understanding that $k_0=0$
forces $k\equiv0$, so that $k_0=0$ characterises the marginally bound case. The
degenerate situation $p=q=+\infty$, i.e.\ homogeneous data, will be excluded by
Assumption~\ref{ass:data}(A1) below, since it gives $t_s'\equiv0$.

The above hypotheses imply \cite{Giambo:2002xc} that the initial data are regular at the center and it is immediate to check that the function $t_s(x)$ \eqref{eq:ts} can be continuously extended up to $x=0$. Also, the \textit{apparent horizon curve} $t_h(x)$, defined for $x>0$ by
$r(t_h(x),x)=2m(x)$, is everywhere smooth on $(0,x_b]$, and 
similarly to the marginally bound case, non central singularities are covered, i.e. $t_h(x)<t_s(x),\,\forall x\in(0,x_b]$, whereas the central shell $x=0$ becomes trapped at the same comoving time it becomes singular, i.e. $t_s(0)= t_h(0)$ so in principle, it can give rise to a naked singularity. 

\begin{remark}
No sign condition is imposed on $k$, which can actually also vanish. 
What is
required, besides $k>-1$, is that every shell actually collapses from rest at
$t=0$, i.e.\ that $r_t$ be real at $r=x$:
\begin{equation}\label{eq:collapsecond}
  2m(x)+k(x)\,x>0,
\quad x\in(0,x_b].
\end{equation}
The above condition places the argument of $\Gamma$
in \eqref{eq:solk} and \eqref{eq:ts} inside its domain of definition $z<1$; in
the limit $x\to0^+$ it reads $2m_0+k_0>0$. \end{remark}

It is convenient to record the rescaled form of \eqref{eq:ts}. Setting
\begin{equation}\label{eq:rescaled}
  \hat m(x):=\frac{m(x)}{x^3},\qquad \hat k(x):=\frac{k(x)}{x^2},
\end{equation}
and inserting the definition \eqref{eq:3.13} of $\Gamma$ into \eqref{eq:ts}, one
gets
\begin{equation}\label{eq:tsrescaled}
  t_s(x)=\sqrt{\frac{x^3}{2m(x)}}\int_0^1
  \frac{\sqrt{s}}{\sqrt{1+\frac{k(x)x}{2m(x)}\,s}}\,\mathrm ds
  =\int_0^1\frac{\sqrt{x^3}\,\sqrt{s}}
  {\sqrt{2m(x)+k(x)\,x\,s}}\,\mathrm ds =\int_0^1\frac{\sqrt{s}}{\sqrt{2\hat m(x)+\hat k(x)\,s}}\,\mathrm ds.
\end{equation}
By \eqref{eq:massprofile}--\eqref{eq:kprofile} the functions $\hat m$ and $\hat k$
are smooth up to $x=0$, with $\hat m(0)=m_0$ and $\hat k(0)=k_0$; hence $t_s$
extends continuously to $x=0$, with $t_s(0)$ obtained from
\eqref{eq:tsrescaled} by replacing $\hat m,\hat k$ with $m_0,k_0$. Note that the
whole $x$-dependence of $t_s$ enters only through $\hat m$ and $\hat k$, a fact
that will be used systematically in Section~\ref{sec:photonsurfaces}.

For our forthcoming analysis it will be also necessary to introduce the parameters
\begin{equation}\label{eq:parameters}
\begin{aligned}
n &=\min\{p,q\},\\
a_n &=
-\frac12\int_0^1 \frac{(2m_n + k_n\tau)\sqrt{\tau}}{(2m_0 + k_0\tau)^{3/2}}\,\mathrm{d}\tau ,
\end{aligned}
\end{equation}
When $n=3$ we simply write $a:=a_3$, which is the parameter used in the
marginally bound analysis of \textsc{Paper~I}.

The integral in \eqref{eq:parameters} is well defined for every sign of $k_0$,
thanks to \eqref{eq:collapsecond}. When $k_0\ne0$ it can be evaluated in closed
form as
\begin{equation}\label{eq:aclosed}
a_n=\frac{1}{\sqrt{2m_0}}\left[\Gamma\!\left(-\frac{k_0}{2m_0}\right)\!
\left(\frac{m_n}{m_0}-\frac{3}{2}\frac{k_n}{k_0}\right)
-\frac{1}{\sqrt{1+\frac{k_0}{2m_0}}}\left(\frac{m_n}{m_0}-\frac{k_n}{k_0}\right)\right],
\end{equation}
both terms being well defined because $-k_0/(2m_0)<1$ by
\eqref{eq:collapsecond}. The singularity of \eqref{eq:aclosed} at $k_0=0$ is
removable: the two divergent contributions proportional to $k_n/k_0$ cancel, and
the limit reproduces the value obtained by evaluating \eqref{eq:parameters}
directly. By the convention above, $k_0=0$ forces $k\equiv0$ and hence $k_n=0$,
so that
\begin{equation}\label{eq:ak0zero}
  a_n\big|_{k_0=0}=-\frac{2\,m_n}{3\,(2m_0)^{3/2}},
\end{equation}
which is the expression used in \textsc{Paper~I}.

Differentiating \eqref{eq:tsrescaled} one obtains the identity that will be used
repeatedly in Section~\ref{sec:photonsurfaces}:
\begin{equation}\label{eq:tsprimeasympt}
  t_s'(x)=n\,a_n\,x^{\,n-1}\bigl(1+o(1)\bigr),\qquad x\to0^+ .
\end{equation}
The estimate \eqref{eq:tsprimeasympt} is sharp provided $a_n\ne0$; the degenerate
case $a_n=0$, which can occur when $p=q$ and the contributions of $m_n$ and $k_n$
cancel, requires expanding to the next order and is not considered here.

In order to have a true inhomogeneous collapse, further conditions must be imposed on the initial data to ensure that $r_x(t,x)>0$
throughout the region where the cloud is regular; therefore, excluding the so called shell-crossing singularity formation. By analogy with the marginally bound case, we require the conditions 
\begin{equation}\label{eq:tsprime}
t_s'(x)> 0, \quad \forall x\in(0,x_b],
\end{equation}
and
\begin{equation}\label{eq:mprime}
    m'(x)>0,\quad\forall x\in(0,x_b].
\end{equation}
Condition \eqref{eq:tsprime} has a direct physical meaning: since $t_s$ is
increasing, the central shell $x=0$ is the first one to become singular, which is
the situation of interest here. By \eqref{eq:tsprimeasympt}, it also forces
$a_n\ge 0$; in the case $n=3$ this means $a\ge0$, so that a part of the
black-hole regime, namely $a<0$, is not covered by the present analysis. As we are showing now, \eqref{eq:tsprime}--\eqref{eq:mprime} are
sufficient conditions for the absence of shell crossings, and models with
$a<0$ can be treated by imposing $r_x>0$ directly.

\begin{remark}
\label{rem:sc}    
Let us verify that these conditions exclude shell-crossing formation for $t>0$. 
Recalling \eqref{eq:LTBevol}, we rewrite equation \eqref{eq:solk} as
\begin{equation}
    \label{eq:sc1}
    t_s(x)-t-\int_0^{r(t,x)}\frac{\mathrm ds}{\sqrt{k(x)+\frac{2m(x)}{s}}}=0,
\end{equation}
and deriving it with respect to $x$ we obtain the following expression for $r_x$:
\begin{equation}
    \label{eq:rx}
    r_x(t,x)=\sqrt{k(x)+\frac{2m(x)}{r(t,x)}}
    \left(
t'_s(x)+\frac12\int_0^{r(t,x)}\frac{k'(x)+\frac{2m'(x)}{s}}{\left(k(x)+\frac{2m(x)}{s}\right)^{3/2}}\,\mathrm ds
    \right)
\end{equation}
The sign of $r_x$ totally depends on the sign of the second factor in the right-hand side above. 
Therefore, let us set 
$$
B(r,x)=t_s'(x)+\frac12\int_0^{r}\frac{k'(x)+\frac{2m'(x)}{s}}{\left(k(x)+\frac{2m(x)}{s}\right)^{3/2}}\,\mathrm ds
$$
and fix $x\in(0,x_b]$. First, observe that  $B(0,x)=t_s'(x)> 0$ by \eqref{eq:tsprime}.  The numerator in the integrand function of $s$ positively  diverges  as $s\to 0^+$, due to \eqref{eq:mprime}, and it may change sign at most once: this means that its integral function in $B(r,x)$ is either always increasing (and then $B(r,x)>0$) or it increases, reaches a maximum and then decreases up to $r$: in this case the minimum of the integral is reached at one of the extrema. But $r(t,x)\le x$ (in particular $r=x$ at $t=0$) so we can say that
$$
B(r,x)
\ge\text{min}\{B(0,x),B(x,x)\}=
\text{min}\left\{ t_s'(x),\ \frac{1}{
\sqrt{k(x)+\frac{2m(x)}{x}}}\right\}> 0
$$
where the last equality is obtained from \eqref{eq:rx}, recalling that $r_x=1$ at
$t=0$. This verifies that $r_x>0$, so that shell-crossing singularity formation is
excluded, on the region $t\ge0$ where the collapse is prescribed.

The same conclusion holds on the whole of $\mathcal D$, including $t<0$, provided
$x_b$ is small enough. Indeed, write $N(s,x):=k'(x)+2m'(x)/s$ for the numerator of
the integrand above, and note that $B(r,x)\ge t_s'(x)>0$ as soon as $N\ge0$ on the
range of integration. By \eqref{eq:kprofile} one has $k'(x)=x\bigl(2k_0+O(x^{\,q})\bigr)$
while $k(x)=x^2\bigl(k_0+O(x^{\,q})\bigr)$, so that for $x$ small the sign of $k'$
is that of $k_0$, and three cases arise.
\begin{itemize}
\item If $k_0>0$, then $k'>0$ and, since $m'>0$ by \eqref{eq:mprime}, $N>0$ for
every $s>0$: no restriction on the range of $r$ is needed.
\item If $k_0<0$, the shell has a maximum expansion radius and $r\le2m/|k|$, so
that $N\ge0$ follows from $|k'|/|k|\le m'/m$; by
\eqref{eq:massprofile}--\eqref{eq:kprofile},
\[
  \frac{m'}{m}-\frac{|k'|}{|k|}=\frac{1+O(x^{\,n})}{x}\;>\;0 .
\]
\item If $k_0=0$, then $k\equiv0$ by the convention after \eqref{eq:kprofile},
$k'=0$ and $N=2m'/s>0$.
\end{itemize}
In all three cases $r_x>0$ on $\mathcal D$, regardless of the sign of $t$.
\end{remark}

We can now collect the hypotheses that will be used throughout the paper.

\begin{assumption}\label{ass:data}
The free functions $m(x)$ and $k(x)$ are as in
\eqref{eq:massprofile}--\eqref{eq:kprofile}, with $k(x)>-1$, and the boundary
label $x_b$ is chosen small enough that, on the region $\mathcal D:=\{(t,x):\ 0<x\le x_b,\ t<t_s(x),\ r_t(t,x)<0\}$,
\begin{enumerate}
  \item[(A1)] $t_s'(x)>0$ and $m'(x)>0$ for all $x\in(0,x_b]$, and $a_n>0$
  in \eqref{eq:parameters};
  \item[(A2)] $r_x(t,x)>0$ on $\mathcal D$, i.e.\ no shell crossing occurs;
  \item[(A3)] $3m(x_b)<x_b$.
\end{enumerate}
\end{assumption}

By Remark~\ref{rem:sc}, (A1) implies (A2) on the whole of $\mathcal D$ for $x_b$
small enough, so that (A2) is listed separately only because it is the property
actually used in the proofs; models satisfying $r_x>0$ for other reasons are
equally admissible. Condition (A3) is what guarantees the
existence of the matching time in Section~\ref{sec:photonsurfaces}: since
$r(0,x_b)=x_b$ and $r_t<0$, it ensures that the boundary shell does cross the
value $r=3m(x_b)$ at some positive comoving time. Finally, $r_x>0$ together with
\eqref{eq:mprime} ensures that the weak energy condition is satisfied.

\subsection{Causal structure of the central singularity}
\label{subsec:causal}

The causal structure of the general LTB model has been studied in
depth~\cite{Singh1996,Deshingkar1999,Mena2001,Joshi2012}. We collect here the
results that will play a role in our analysis, as recast in the framework of \cite{Giambo2003}, where a general class of spacetimes, including the present case as a particular situation, was analyzed. There, a class of spherically symmetric spacetimes is considered,
\begin{equation}
  g = -e^{2\nu}\,dt^2 + e^{2\lambda}\,dx^2 + r^2\,d\Omega^2,
  \label{eq:generalmetric}
\end{equation}
with $\lambda,\nu,r$ functions of $(t,x)$,
where the function
\begin{equation}\label{eq:ru2}
r\, (r_t e^{-\nu})^2=2m+r\,\left[(r_x e^{-\lambda})^2-1\right],
\end{equation}
can be seen as a function of $(x,r)$; therefore, it includes dependence on $t$ through $r(t,x)$. In this case, the final state of the central singularity depends on the Taylor expansion of \eqref{eq:ru2} in the variables $(x,r)$. In our case, it takes the form $2m(x)+r k(x)$. 

In view of this, the central singularity may possibly become naked,
yielding the following
classification \cite{Giambo:2002xc,Giambo2003}: 
recalling \eqref{eq:parameters}
we have that:
\begin{itemize}
  \item the central singularity is \emph{naked} if either $n\in\{1,2\}$, or
  $n=3$ with the coefficient $a=a_3$ in~\eqref{eq:parameters} suitably large: $a>a_{\text{c}}:=\tfrac{26+15\sqrt3}{3}m_0.$
  \item It is \emph{covered} (i.e.\ a black-hole singularity) otherwise,
  i.e.\ for $n\ge 4$ or for $n=3$ with $a\le a_{\text{c}}$.
\end{itemize}
For $k\equiv0$ this reduces to the threshold of \textsc{Paper~I}: the normalisation $t_s(0)=1$ used there corresponds to $m_0=2/9$, whence $a_{\mathrm c}=\tfrac2{27}(26+15\sqrt3)$.

The causal classification therefore retains the same form as in the marginally
bound case, the energy function entering only through the value of $a_n$. The
photon surface equations, by contrast, are genuinely modified, as we shall see in
Section~\ref{sec:photonsurfaces}: whether the two continue to correspond is the question addressed in the rest of this paper.

A crucial technical tool, which we shall repeatedly use in the sequel, is the
following comparison property~\cite{Giambo2003}: the apparent
horizon curve $t_h(x)$ is a \emph{subsolution} of the outgoing radial null
geodesic equation
\begin{equation}
  t'(x) = \frac{r_x(t(x),x)}{\sqrt{1+k(x)}}, \qquad x\in(0,x_b),
  \label{eq:nullgeodesic}
\end{equation}
in the sense that $t_h'(x)\le r_x(t_h(x),x)/\sqrt{1+k(x)}$. This property,
which holds for the general LTB metric exactly as in the marginally bound
case~\cite{GiamboLucamarini2024} is the cornerstone of the analysis needed in order to obtain the above classification: we refer the reader to \cite{Giambo:2002xc,Giambo2003} for the proof. We will further comment on this in
Section~\ref{sec:analysis}.

\section{Photon surfaces in dust collapse: the general case}
\label{sec:photonsurfaces}

\subsection{General dynamical-systems formulation}

We briefly recall the framework developed in \textsc{Paper~I}~\cite{GiamboLucamarini2024};
the reader is referred to that paper for full derivations and a complete
discussion. For a general spherically symmetric metric of the form \eqref{eq:generalmetric}, the condition that an
$SO(3)$-invariant timelike hypersurface
$\mathcal{S}=\{(t,x(t),\theta,\varphi)\}$ be a photon surface in the sense of
Claudel--Virbhadra--Ellis~\cite{Claudel2001} is equivalent (cf. the corresponding Proposition of
\textsc{Paper~I}, and the quasi-local characterisation
of~\cite{Cao2021}) to the first-order non-autonomous dynamical system
\begin{align}
  \dot{x}(t) &= e^{\nu-\lambda}\,y(t), \label{eq:gendyn1} \\
  \dot{y}(t) &= \bigl(1-y(t)^2\bigr)\,\left[e^{\nu-\lambda}\,\frac{r_x - r\,\nu_x}{r}
              + y(t)\,\frac{r_t - r\,\lambda_t}{r}\right], \label{eq:gendyn2}
\end{align}
where $y(t)\in[-1,1]$ is an auxiliary variable encoding the direction of the
tangent vector to the hypersurface relative to the comoving observer. The
boundary values $y=\pm 1$ are equilibrium points of~\eqref{eq:gendyn2}, and they
correspond to outgoing ($y=+1$) and ingoing ($y=-1$) radial null geodesics,
which are trivially photon surfaces (being null hypersurfaces). The intermediate
values $y\in(-1,1)$ correspond to genuinely timelike photon surfaces tilted with
respect to the comoving frame.

\subsection{Specialisation to the LTB metric with $k(x)\ne 0$}

For the LTB metric~\eqref{eq:LTBmetric} we read off
$e^{\nu}=1$ and $e^{\lambda}=r_x/\sqrt{1+k(x)}$, hence
$e^{\nu-\lambda}=\sqrt{1+k(x)}/r_x$. 
The corresponding strictly equivalent
first-order system reads:
\begin{align}
  \dot{x}(t) &= \frac{\sqrt{1+k(x)}}{r_x}\,y(t), \label{eq:LTBdyn1} \\[4pt]
  \dot{y}(t) &= \bigl(1-y(t)^2\bigr)\,\left[\frac{\sqrt{1+k(x)}}{r}
              + y(t)\,\left(\frac{r_t}{r} -\frac{r_{tx}}{r_x}\right)\right].
  \label{eq:LTBdyn2}
\end{align}
For $k\equiv 0$ the system~\eqref{eq:LTBdyn1}--\eqref{eq:LTBdyn2} reduces to the corresponding system of \textsc{Paper~I}. The factor $\sqrt{1+k(x)}$ in~\eqref{eq:LTBdyn1}--\eqref{eq:LTBdyn2} is a genuine
effect of the non-marginally bound regime, and it will turn out not to affect the leading
balance which governs the behaviour of the photon surface near the central
singularity. The features established in Sections~\ref{subsec:key}--\ref{subsec:nullext} are therefore not consequences of $k\ne0$.

\subsection{Initial conditions and admissibility}
\label{subsec:initialdata}

The exterior photon surface, in the Schwarzschild geometry, is located at $r=3M$.
Its extension into the LTB interior must therefore start from the boundary
$\Sigma=\{x=x_b\}$ at the comoving time $t_0>0$ defined implicitly by
\begin{equation}
  r(t_0,x_b) = 3\,m(x_b).
  \label{eq:t0def}
\end{equation}
Such a time exists and is unique. 
Indeed, $r(0,x_b)=x_b>3m(x_b)$ by Assumption~\ref{ass:data}(A3), while
$r(t,x_b)\to0$ as $t\to t_s(x_b)^-$; since $r_t<0$, the map
$t\mapsto r(t,x_b)$ is a strictly decreasing bijection from $[0,t_s(x_b))$ onto
$(0,x_b]$, and \eqref{eq:t0def} has exactly one solution $t_0\in(0,t_s(x_b))$.
Note that $r(t_0,x_b)=3m(x_b)>2m(x_b)$, so that $t_0<t_h(x_b)$: the matching
point lies outside the apparent horizon, as it must.

The second piece of initial data for system~\eqref{eq:LTBdyn1}--\eqref{eq:LTBdyn2}
is the value $y_0=y(t_0)$, which encodes the tilt of the photon surface at the
matching point. The requirement that the surface be non-spacelike at $t_0$ forces
$|y_0|\le 1$. The boundary values $y_0=\pm1$ are equilibria of
\eqref{eq:LTBdyn2}, so that the trajectory remains on the null hypersurface
$y\equiv\pm1$ for all $t$; the case $y_0=-1$ is the ingoing radial null geodesic,
which moves towards smaller $r$ and therefore cannot extend the exterior photon
sphere backwards in time.

The values $y_0\le0$ can be discarded at once. If $y_0<0$ then $\dot x(t_0)<0$ by
\eqref{eq:LTBdyn1}, so $x(t)>x_b$ for $t<t_0$ and the trajectory leaves the cloud
immediately. If $y_0=0$ then $\dot x(t_0)=0$ and, by \eqref{eq:LTBdyn2},
$\dot y(t_0)=\sqrt{1+k(x_b)}/r(t_0,x_b)>0$, so that $y(t)<0$ for $t<t_0$ and the
same conclusion holds. The only case that has to be discussed is therefore
$y_0\in(0,1)$, and the question is whether some such value yields a geometrically
and physically acceptable timelike extension.

Throughout the rest of this section we shall use $x$, rather than $t$, as a
parameter along the trajectory, which is legitimate as long as $y(t)>0$, since in that
case $\dot x(t)>0$ by \eqref{eq:LTBdyn1}, so $x(t)$ is strictly increasing and
invertible, and the inverse $t(x)$ satisfies
\begin{equation}\label{eq:ODEt}
  \frac{\mathrm dt}{\mathrm dx}=\frac{r_x(t(x),x)}{\sqrt{1+k(x)}\;y(t(x))}\;>\;0 .
\end{equation}
We stress that, along such a trajectory, the partial derivative $r_x$ and the
total derivative of $r(t(x),x)$ with respect to $x$ do not coincide. Rather,
\begin{equation}\label{eq:totalvspartial}
  \frac{\mathrm d}{\mathrm dx}\ln r(t(x),x)
  =\frac{r_x}{r}+\frac{r_t}{r}\,\frac{\mathrm dt}{\mathrm dx}
  \;<\;\frac{r_x}{r},
\end{equation}
the inequality following from $r_t<0$ and $\mathrm dt/\mathrm dx>0$. In
particular, for $0<x_a<x_b$,
\begin{equation}\label{eq:logbound}
  \int_{x_a}^{x_b}\frac{r_x}{r}\,\mathrm dx\;\ge\;
  \ln\frac{r(t(x_b),x_b)}{r(t(x_a),x_a)} ,
\end{equation}
an inequality which clarifies the difference between the partial derivative $r_x$
and the total derivative of $r$ along the trajectory.

We first establish that any such extension remains outside the apparent horizon.
This will allow us to use the inequality $2m/r<1$ in the reduction of
Section~\ref{subsec:key} and, as an upper barrier for $\chi$, in the proof of Proposition~\ref{prop:naked}.

\begin{lemma}\label{lem:hor}
Let Assumption~\ref{ass:data}(A1)--(A3) hold and let $(x(t),y(t))$ be a solution of
system~\eqref{eq:LTBdyn1}--\eqref{eq:LTBdyn2} with initial data
$\bigl(x(t_0),y(t_0)\bigr)=(x_b,y_0)$, $y_0\in(0,1]$. Then
\[
  2m\bigl(x(t)\bigr)<r\bigl(t,x(t)\bigr)
\]
for every $t\le t_0$ such that $y(\tau)\in(0,1]$ for all $\tau\in[t,t_0]$.
\end{lemma}

\begin{proof}
As a preliminary step, let us determine the differential equation satisfied by
$t_h(x)$, defined for $x>0$ by
\begin{equation}\label{eq:xhdef}
  r(t_h(x),x)=2m(x).
\end{equation}
Differentiating this relation with respect to $x$ and using \eqref{eq:LTBevol}
and \eqref{eq:xhdef}, so that $r_t(t_h(x),x)=-\sqrt{1+k(x)}$, we get
\begin{equation}\label{eq:ODEh}
  \frac{\mathrm dt_h}{\mathrm dx}=\frac{r_x(t_h(x),x)-2m'(x)}{\sqrt{1+k(x)}}.
\end{equation}
Since $y(t)>0$, the inverse function $t(x)$ is well defined and satisfies
\eqref{eq:ODEt}. Comparing the two vector fields at any common point $(t,x)$,
and using $m'>0$, $r_x>0$ and $0<y\le1$ from Assumption~\ref{ass:data}, we obtain
\[
  \frac{r_x(t,x)-2m'(x)}{\sqrt{1+k(x)}}
  \;<\;\frac{r_x(t,x)}{\sqrt{1+k(x)}}
  \;\le\;\frac{r_x(t,x)}{\sqrt{1+k(x)}\,y}.
\]
Moreover $t(x_b)=t_0<t_h(x_b)$, as observed after \eqref{eq:t0def}. 
The conclusion follows from comparison theorems in ODE, integrating backwards in the variable $x$.
\end{proof}

\begin{remark}\label{rem:chi}
Equivalently, by \eqref{eq:LTBevol}, Lemma~\ref{lem:hor} states that
\begin{equation}\label{eq:chidef}
  \chi:=\frac{-r_t}{\sqrt{1+k}}=\frac{1}{\sqrt{1+k}}\sqrt{k+\frac{2m}{r}}\;<\;1
\end{equation}
along the trajectory. The quantity $\chi$ measures how close the trajectory is to
the apparent horizon, which is the locus $\chi=1$, and it will play the role of a
dynamical variable in Section~\ref{subsec:nullext}. At the matching point one has,
by \eqref{eq:t0def},
\begin{equation}\label{eq:chi0}
  \chi(t_0)^2=\frac{k(x_b)+\tfrac23}{1+k(x_b)}\;\xrightarrow[x_b\to0^+]{}\;\frac23 ,
\end{equation}
a value which is independent of the model and 
will be important.
\end{remark}

Lemma~\ref{lem:hor} confines the trajectory to the untrapped region, but it does
not by itself decide whether an intermediate value $y_0\in(0,1)$ is admissible.
This is the object of the next subsection, where we shall see that the answer
depends on the causal nature of the central singularity.

\subsection{Scalar reduction}
\label{subsec:key}

Along a trajectory with $0<y<1$ it is convenient to replace $y$ by the variable
\begin{equation}\label{eq:udef}
  u:=-\tfrac12\ln\bigl(1-y^2\bigr)\in[0,+\infty),
\end{equation}
so that $u=0$ if and only if $y=0$, while $u\to+\infty$ if and only if $y\to1$.
Dividing \eqref{eq:LTBdyn2} by \eqref{eq:LTBdyn1} and multiplying by
$y/(1-y^2)$, the whole system collapses to the single scalar equation
\begin{equation}\label{eq:ustar}
  \frac{\mathrm du}{\mathrm dx}=\frac{r_x}{r}+\frac{y}{\sqrt{1+k}}\,D,
  \qquad\text{where}\qquad
  D:=\frac{r_x\,r_t}{r}-r_{tx},
\end{equation}
all functions being evaluated at $(t(x),x)$. Since $u\ge0$ by construction, while
$u(x_b)=u_0:=-\tfrac12\ln(1-y_0^2)$ is finite, we obtain the basic inequality
\begin{equation}\label{eq:basic}
  0\;\le\;u(x)\;=\;u_0-\int_{x}^{x_b}\frac{r_x}{r}\,\mathrm dx
  \;-\;\int_{x}^{x_b}\frac{y}{\sqrt{1+k}}\,D\,\mathrm dx ,
  \qquad 0<x<x_b .
\end{equation}
The first integral is positive and, by \eqref{eq:logbound}, diverges when
$r(t(x),x)\to0$. The entire question of whether a trajectory with $y_0\in(0,1)$
can reach the centre without $y$ ever vanishing is therefore reduced to understand how large can the second integral become and this is what we estimate
next.

It is useful to derive an exact expression for $D$ in terms of the rescaled
quantities \eqref{eq:rescaled}. Exactly as in the derivation of
\eqref{eq:tsrescaled}, the substitution $s=x\sigma$ turns the implicit solution
\eqref{eq:solk} into
\begin{equation}\label{eq:Fdef}
  t=F(1,x)-F(S,x),\qquad
  F(S,x):=\int_0^S\frac{\sqrt\sigma}{\sqrt{2\hat m(x)+\hat k(x)\sigma}}\,\mathrm d\sigma,
  \qquad S:=\frac{r(t,x)}{x},
\end{equation}
so that $F(1,x)=t_s(x)$ by \eqref{eq:tsrescaled}. Setting moreover
\begin{equation}\label{eq:Sdef}
  \omega(S,x):=\sqrt{\hat k(x)+\frac{2\hat m(x)}{S}}=-\frac{r_t}{x},\qquad
  \Phi(S,x):=-F_x(S,x)=\frac12\int_0^S
  \frac{\bigl(2\hat m'(x)+\hat k'(x)\sigma\bigr)\sqrt\sigma}
  {\bigl(2\hat m(x)+\hat k(x)\sigma\bigr)^{3/2}}\,\mathrm d\sigma ,
\end{equation}
one has $\Phi(1,x)=-t_s'(x)$ and $F_S(S,x)=1/\omega(S,x)$. Differentiating
\eqref{eq:Fdef} with respect to $x$ at fixed $t$ gives
$F_x(1,x)-F_S(S,x)\,S_x-F_x(S,x)=0$, whence
\begin{equation}\label{eq:rxS}
  \frac{r}{x}-r_x=-x\,S_x,\qquad
  S_x=\bigl[\Phi(S,x)-\Phi(1,x)\bigr]\,\omega(S,x),
\end{equation}
the first identity being immediate from $r=xS$. 

Substituting $k=x^2\hat k$, $m=x^3\hat m$, $r=xS$ and $r_t=-x\,\omega$ into the
definition of $D$ in \eqref{eq:ustar} one obtains, after simplification,
\begin{equation}\label{eq:Dmid}
  D=\frac{\bigl(S-r_x\bigr)\bigl(3\hat m+\hat k S\bigr)}{S^2\,\omega}
   +\frac{x\bigl(2\hat m'+\hat k'S\bigr)}{2\,S\,\omega},
\end{equation}
and inserting $S-r_x=x\bigl[\Phi(1,x)-\Phi(S,x)\bigr]\omega$ from
\eqref{eq:rxS}, the factor $\omega$ cancels in the first term, leaving
\begin{equation}\label{eq:Dexact}
  D=\frac{x\bigl[\Phi(1,x)-\Phi(S,x)\bigr]\bigl(3\hat m+\hat k S\bigr)}{S^2}
   +\frac{x\bigl(2\hat m'+\hat k'S\bigr)}{2\,S\,\omega}.
\end{equation}
In particular, if $\hat m$ and $\hat k$ are constant, that is if the collapse is
homogeneous, then $\Phi\equiv0$ and $D\equiv0$ identically: the term
$D$ in \eqref{eq:ustar} is a pure inhomogeneity effect.

For the analysis of Lemma~\ref{lem:main} it is convenient to consider a second, and
equivalent, form of \eqref{eq:ustar}, in which the possible degeneracy of $r_t$ is
isolated. By \eqref{eq:LTBevol} one has $r_t=-\chi\sqrt{1+k}$ with $\chi$ as in
\eqref{eq:chidef}, and
\[
  r_{tx}=-\frac{1}{2\,\chi\sqrt{1+k}}
  \left(k'+\frac{2m'}{r}-\frac{2m}{r}\,\frac{r_x}{r}\right).
\]
Substituting into $D=r_xr_t/r-r_{tx}$, and eliminating $m/r$ through the identity
$2m/r=\chi^2(1+k)-k$, the term proportional to $r_x/r$ can be collected with the
first one, giving the exact relation
\begin{equation}\label{eq:ustar2}
  \frac{\mathrm du}{\mathrm dx}
  =\frac{r_x}{r}\,\beta\;+\;\frac{y}{\chi}\,\frac{k'+\frac{2m'}{r}}{2\,(1+k)},
  \quad\text{where}\,\,
  \beta:=1-\frac32\,\chi\,y+\frac{y\,\kappa}{2\,\chi},
  \quad \kappa:=\frac{k}{1+k}.
\end{equation}
The interest of \eqref{eq:ustar2} is that $r_x$ no longer appears in the
singular term: the factor $1/\chi$, which diverges where $r_t$ vanishes, is
multiplied by $y$, and $y/\chi$ turns out to be bounded precisely in the regime
in which that degeneracy can be approached. 

The following estimate shows that,
away from the singularity, the term carrying $D$ is integrable and therefore
cannot compete with the logarithmic divergence of the first integral in  \eqref{eq:basic}.

\begin{lemma}\label{lem:Dbound}
Let Assumption~\ref{ass:data} hold and let $c>0$. Then there exist $C>0$ and
$x_*\in(0,x_b]$ such that 
\begin{equation}\label{eq:Dbound}
  \bigl|D(t,x)\bigr|\le C\,x^{\,n}\bigl(1+S^{-2}\bigr)
\end{equation}
whenever $0<x\le x_*$ and $2\hat m(x)+\hat k(x)\,S\ge c$, the latter being
equivalent to $\chi^2(1+k)\,S\ge c\, x^2$. In particular, if moreover
$S\ge S_{\min}>0$, then $|D|\le C'x^{\,n}$.
\end{lemma}

\begin{proof}
By \eqref{eq:massprofile}--\eqref{eq:kprofile} we have $\hat m'(x)=O(x^{\,n-1})$
and $\hat k'(x)=O(x^{\,n-1})$, uniformly for $x$ small. Moreover, since
$\sigma\mapsto2\hat m(x)+\hat k(x)\sigma$ is affine, its minimum over $[0,S]$ is
attained at an endpoint, so that the hypothesis $2\hat m+\hat kS\ge c$, together
with $2\hat m(x)\to2m_0>0$, gives
\[
  2\hat m(x)+\hat k(x)\,\sigma\;\ge\;\min\{2\hat m(x),\,c\}\;\ge\;c_0>0
  \qquad\text{for every }\sigma\in(0,S]
\]
and every $x$ small: the denominator in \eqref{eq:Sdef} is bounded below on the
whole range of integration. Fix now $S_0>0$.

For $S\le S_0$, the integrand in \eqref{eq:Sdef} is $O(x^{\,n-1}\sqrt\sigma)$,
whence $\Phi(S,x)=O(x^{\,n-1})$ and likewise $\Phi(1,x)=O(x^{\,n-1})$. Since
$3\hat m+\hat kS$ is bounded, the first term in \eqref{eq:Dexact} is
$O(x^{\,n}/S^2)$; and since
$\omega=\sqrt{(2\hat m+\hat kS)/S}\ge\sqrt{c_0}\,S^{-1/2}$, the second one is
$O(x^{\,n}S^{-1/2})$. Both are $O\bigl(x^{\,n}(1+S^{-2})\bigr)$, because
$S^{-1/2}\le S^{-2}$ for $S\le1$ and $S^{-1/2}\le1$ for $S\ge1$.

For $S\ge S_0$ we distinguish two subcases. If $\hat k\ne0$, then necessarily
$\hat k>0$ as soon as $S_0$ is large enough, since $\hat k<0$ together with
$2\hat m+\hat kS\ge c$ would bound $S$; the integrand in \eqref{eq:Sdef} then
tends to the constant $\hat k'\hat k^{-3/2}=O(x^{\,n-1})$, so that
$\Phi(S,x)=O(x^{\,n-1}S)$, while $3\hat m+\hat kS=O(S)$ and
$\omega\ge\sqrt{\hat k}$, and both terms in \eqref{eq:Dexact} are $O(x^{\,n})$.
If instead $\hat k=0$, then $k\equiv0$ and $\hat k'=0$ by the convention
following \eqref{eq:kprofile}, so that
$\Phi(S,x)=\tfrac23\,\hat m'\,(2\hat m)^{-3/2}S^{3/2}=O(x^{\,n-1}S^{3/2})$, while
$3\hat m+\hat kS=O(1)$ and $\omega=\sqrt{2\hat m/S}$; both terms in
\eqref{eq:Dexact} are then $O(x^{\,n}S^{-1/2})$, hence again $O(x^{\,n})$. In
either subcase this is the first component in \eqref{eq:Dbound}.
\end{proof}

In the covered case the trajectory cannot approach the central singularity, which
is what makes the hypothesis of Lemma~\ref{lem:Dbound} available. The following
comparison makes this precise.

\begin{lemma}\label{lem:Smin}
Let Assumption~\ref{ass:data} hold, assume that the central singularity is
covered, and let $(x(t),y(t))$ be a solution of
\eqref{eq:LTBdyn1}--\eqref{eq:LTBdyn2} with $\bigl(x(t_0),y(t_0)\bigr)=(x_b,y_0)$,
$y_0\in(0,1]$. Then there exists $S_{\min}>0$ such that $S=r/x\ge S_{\min}$ at
every point of the backward trajectory at which $y>0$.
\end{lemma}

\begin{proof}
Since $0<y\le1$, equation \eqref{eq:ODEt} gives
$\mathrm dt/\mathrm dx\ge r_x/\sqrt{1+k}$: the trajectory is a supersolution of
the outgoing radial null geodesic equation \eqref{eq:nullgeodesic}. Denoting by
$t_N$ the solution of \eqref{eq:nullgeodesic} with $t_N(x_b)=t_0$, integration
backwards from $x_b$ yields $t(x)\le t_N(x)$, and since $r_t<0$,
\begin{equation}\label{eq:comparisonS}
  r(t(x),x)\;\ge\;r(t_N(x),x),\qquad\text{that is}\qquad S(x)\ge S_N(x).
\end{equation}
Because the singularity is covered, the results of \cite{Giambo2003} imply that
$t_N$ extends backwards to the regular centre, i.e.\ $t_N(x)\to t_*<t_s(0)$ as
$x\to0^+$. Consequently $S_N(x)$ tends to the positive value $S_*$ determined by
$t_s(0)-t_*=\int_0^{S_*}\!\sqrt\sigma\,(2m_0+k_0\sigma)^{-1/2}\,\mathrm d\sigma$,
so that $S_N$ is bounded below by a positive constant on $(0,x_b]$; the claim
follows from \eqref{eq:comparisonS}.
\end{proof}

The second integral in \eqref{eq:basic} is thus under control. The following  result is the crucial tool to prevent the existence of timelike extensions of photon surfaces, and
as in \textsc{Paper~I}, it requires the central singularity to be covered.

\begin{lemma}
\label{lem:main}
Let Assumption~\ref{ass:data} hold and assume that the central singularity is
covered, i.e.\ $n\ge4$, or $n=3$ with $a\le a_{\mathrm c}$. Let $(x(t),y(t))$ be
a solution of \eqref{eq:LTBdyn1}--\eqref{eq:LTBdyn2} with initial data
$\bigl(x(t_0),y(t_0)\bigr)=(x_b,y_0)$, $y_0\in(0,1)$. Then there exists
$\tilde t<t_0$ such that $x(\tilde t)=x_b$.
\end{lemma}

\begin{proof}
\emph{Step 1.} We first claim that
\begin{equation}\label{eq:P1}
  \exists\,t_1<t_0\ :\ y(t_1)=0,\quad x(t_1)>0 .
\end{equation}
We first make an observation which will be used twice below. Whenever the areal
radius increases along the backward evolution, that is whenever
$\mathrm dr/\mathrm dt=\sqrt{1+k}\,(y-\chi)\le 0$, one has $y\le\chi$; in that case
the factor $1/\chi$ carried by
\[
  r_{tx}=-\frac{1}{2\chi\sqrt{1+k}}
  \left(k'+\frac{2m'}{r}-\frac{2m}{r}\,\frac{r_x}{r}\right)
\]
is compensated by the factor $y$ multiplying $r_{tx}/r_x$ in
\eqref{eq:LTBdyn2}, since $y/\chi\le1$. Then the right-hand side of
\eqref{eq:LTBdyn2} stays bounded even where $\chi\to0$.

Let now $(t_{\inf},t_0]$ be the maximal backward interval on which $y>0$. There
$\dot x>0$ by \eqref{eq:LTBdyn1}, so that $x(t)$ is strictly increasing and admits
a limit $x_{\inf}\ge0$ as $t\to t_{\inf}^+$. We claim that either \eqref{eq:P1}
holds, or
\begin{equation}\label{eq:xlimit}
  \lim_{t\to t_{\inf}^+}x(t)=0 .
\end{equation}
Indeed, suppose $x_{\inf}>0$. Since $x$ is strictly increasing it cannot return to
$\{x=x_b\}$, and by Lemma~\ref{lem:hor} the trajectory stays in $\{r>2m\}$, hence
away from $\{t=t_s(x)\}$; three possibilities are therefore left.
\begin{itemize}
\item The limit point $(t_{\inf},x_{\inf})$ is interior to $\mathcal D$. Then all
the coefficients in \eqref{eq:LTBdyn1}--\eqref{eq:LTBdyn2} are continuous and
bounded there, so the solution extends up to $t_{\inf}$; by maximality of the
interval, $y(t_{\inf})=0$, and since $x(t_{\inf})=x_{\inf}>0$ this is
\eqref{eq:P1}.
\item The limit point lies on the rest curve $\{r_t=0\}$, so that $\chi\to0$ and
$r\to2m/|k|$ increases along the backward evolution. By the observation above
$y\le\chi\to0$ and the right-hand side of \eqref{eq:LTBdyn2} stays bounded, so
again the solution extends up to $t_{\inf}$ with $y(t_{\inf})=0$ and
$x(t_{\inf})>0$: once more \eqref{eq:P1}.
\item $t_{\inf}=-\infty$. Then $r\to\infty$ and the second term in
\eqref{eq:LTBdyn2} is negligible with respect to the first, so that
$\dot y\simeq(1-y^2)\sqrt{1+k}/r$, with $r\simeq\sqrt{k(x_{\inf})}\,|t|$ if
$k(x_{\inf})>0$ and $r\simeq\bigl(2m(x_{\inf})\bigr)^{1/3}\bigl(3|t|/2\bigr)^{2/3}$
if $k(x_{\inf})=0$; in both cases the integral of $1/r$ diverges as $t\to-\infty$, so that
$y$ would vanish at some finite time of $(t_{\inf},t_0]$, against the definition
of that interval.
\end{itemize}
This proves the claim. It remains to exclude \eqref{eq:xlimit}, and we do so by
contradiction: assume that $x(t)\to0$, and distinguish two regimes according to
whether $\chi$ stays away from zero.

\emph{(a) $\chi\ge\chi_0>0$ along the trajectory.} Then
$2m+k\,r=r\,\chi^2(1+k)\ge\chi_0^2\,x^3 S$, so that the hypothesis of
Lemma~\ref{lem:Dbound} holds with $c=\chi_0^2S_{\min}$, $S_{\min}$ being given by
Lemma~\ref{lem:Smin}. Hence $|D|\le C'x^{\,n}$ and, setting
$C'':=C'\sup_{(0,x_b]}\bigl(1+k(x)\bigr)^{-1/2}$, which is finite because
$k>-1$,
\begin{equation}\label{eq:badterm}
  \left|\int_x^{x_b}\frac{y}{\sqrt{1+k}}\,D\,\mathrm dx'\right|
  \;\le\;\frac{C''x_b^{\,n+1}}{n+1},
\end{equation}
and \eqref{eq:basic} gives
\begin{equation}\label{eq:uestimate}
0\le u(x)\le u_0+\frac{C''x_b^{\,n+1}}{n+1}-\int_x^{x_b}(r_x/r)\,\mathrm dx'.
\end{equation}
If $t_{\inf}$ is finite,
then $r(t(x),x)\to r(t_{\inf},0)=0$ by regularity of the centre, and
\eqref{eq:logbound} makes the right-hand side in \eqref{eq:uestimate} diverge to $-\infty$: a
contradiction.
If $t_{\inf}=-\infty$, which requires $k\ge0$, then $S\to\infty$ and, by the two
subcases in the proof of Lemma~\ref{lem:Dbound}, $\Phi(S,x)\,\omega(S,x)
=O(x^{\,n-1})S$ in either case; hence \eqref{eq:rxS} gives
$r_x=S\bigl(1+O(x^{\,n})\bigr)$, so that $r_x/r\simeq1/x$ and the integral in
\eqref{eq:uestimate} diverges again.

\emph{(b) $\chi\to0$ along the trajectory.} By \eqref{eq:chidef},
$\chi^2(1+k)=k+2m/r$, so $\chi\to0$ forces $k<0$ and $r\to2m/|k|$: the trajectory
approaches the rest curve $\{r_t=0\}$, and $r$ increases along the backward
evolution, so that $y\le\chi$ by the observation made at the beginning of the
proof. Then, in \eqref{eq:ustar2},
\[
  \Bigl|\tfrac32\chi y\Bigr|\le\tfrac32\chi^2\longrightarrow0,
  \qquad
  \Bigl|\frac{y\,\kappa}{2\chi}\Bigr|\le\frac{|k(x)|}{2\bigl(1+k(x)\bigr)}
  \longrightarrow0,
\]
so that $\beta\ge\tfrac12$ for $x$ small. Moreover $r\le2m/|k|$ throughout this
regime, whence $2m'/r\le m'|k|/m=O(x)$ by
\eqref{eq:massprofile}--\eqref{eq:kprofile}, while $k'=O(x)$; since $y/\chi\le1$,
the second term of \eqref{eq:ustar2} is bounded in modulus by $C_1\,x$, for a suitable $C_1>0$.
Therefore
\[
  \frac{\mathrm du}{\mathrm dx}\;\ge\;\frac12\,\frac{r_x}{r}\;-\;C_1\,x
\]
at every point at which $y\le\chi$. Since $\chi\to0$ forces $r\to2m/|k|$, and
hence $r$ to increase along the backward evolution, this holds on an interval
$(0,x_1)$ for some $x_1\in(0,x_b)$. Moreover
$r\le2m(x)/|k(x)|=2x\,\hat m/|\hat k|=O(x)$, so that $r\to0$ as $x\to0$.
Integrating the above on $(x,x_1)$ and using \eqref{eq:totalvspartial},
\[
  u(x_1)-u(x)\;\ge\;\frac12\,\ln\frac{r(t(x_1),x_1)}{r(t(x),x)}
  -\frac{C_1\,x_1^2}{2}\;\xrightarrow[x\to0^+]{}\;+\infty ,
\]
again contradicting $u\ge0$. In either regime \eqref{eq:xlimit} is impossible and  \eqref{eq:P1} follows.

\medskip
\emph{Step 2.} From now on we argue as in the corresponding lemma of \textsc{Paper~I}. The
function $x(t)$ satisfies the second order ODE
\begin{equation}
\begin{split}
  \ddot{x}(t) &= \frac{k(x)+1}{r\,r_x}
              + \left(\frac{r_t}{r} - \frac{2r_{tx}}{r_x}\right)\dot{x}\\
             &\quad - \left(\frac{r_x}{r} + \frac{r_{xx}}{r_x}
                     - \frac{k'(x)}{2\bigl(k(x)+1\bigr)}\right)\dot{x}^2
              + \frac{r_x\bigl(r\,r_{tx} - r_x\,r_t\bigr)}{r\bigl(k(x)+1\bigr)}\,\dot{x}^3 .
\end{split}
\label{eq:LTBODE}
\end{equation}
At $t=t_1$ we have $\dot x(t_1)=0$, so \eqref{eq:LTBODE} reduces to
$\ddot x(t_1)=\bigl(k(x)+1\bigr)/(r\,r_x)>0$: every critical point of $x(t)$ is a
strict local minimum. Hence $t_1$ is a minimum, $y(t)<0$ for $t<t_1$, and $x(t)$
increases as $t$ decreases; it cannot have a further critical point, since that
would be a maximum. Therefore $x(t)$ increases monotonically for $t<t_1$ and
reaches $x_b$ at some $\tilde t<t_1<t_0$.
\end{proof}

\subsection{The null extension theorem}
\label{subsec:nullext}

Lemma~\ref{lem:main} has the following immediate consequence in the covered case.
Suppose the central singularity is covered and let $(x(t),y(t))$ be an extension
of the exterior photon surface with $y_0\in(0,1)$. By Lemma~\ref{lem:main} the
trajectory re-intersects $\Sigma=\{x=x_b\}$ at some earlier comoving time
$\tilde t<t_0$. At that time, however,
\begin{equation}\label{eq:recross}
  r(\tilde t,x_b)>r(t_0,x_b)=3M,
\end{equation}
because the collapse is monotonically ingoing. Hence the hypersurface leaves the
cloud at a point of $\Sigma$ which does not belong to the exterior photon sphere,
and cannot be matched to it. Together with the discussion of $y_0\le0$ and
$y_0=-1$ given in Section~\ref{subsec:initialdata}, this rules out every value
except $y_0=+1$.

When the central singularity is naked the conclusion of Lemma~\ref{lem:main}
fails. Let us see why, analysing the
behaviour of the system near the centre. Let
\begin{equation}\label{eq:chiLambda}
  \chi:=\frac{-r_t}{\sqrt{1+k}},\qquad s:=\ln x,
\end{equation}
so that going backwards along the trajectory corresponds to $s\to-\infty$, and
$\chi<1$ by Lemma~\ref{lem:hor}. Using $\chi^2=2m/r+O(x^2)$, together with
\eqref{eq:rxS}, \eqref{eq:Dexact} and \eqref{eq:tsprimeasympt}, a direct
computation gives, in the limit $x\to0^+$,
\begin{align}
  \frac{\mathrm d\chi}{\mathrm ds}
    &=\chi\left[1-\frac{\Lambda}{2}+\frac{\chi\,(1+\Lambda)}{2\,y}\right]
       +o(1), \label{eq:sys1}\\[4pt]
  \frac{\mathrm du}{\mathrm ds}
    &=1+\Lambda\left(1-\frac{3}{2}\,\chi\,y\right)+o(1), \label{eq:sys2}
\end{align}
where 
\begin{equation}\label{eq:Lambdachi}
  \Lambda:=\frac{x\,r_x}{r}-1=A(x)\,\chi^3,\qquad A(x):=\frac{t_s'(x)}{2m_0\,x^2}
  =\frac{n\,a_n}{2m_0}\;x^{\,n-3}\bigl(1+o(1)\bigr).
\end{equation}
These expansions are valid in the regime relevant here, namely along trajectories
approaching the central singularity, for which $S=r/x\to0$ and $\chi$ stays
bounded away from zero; equivalently, $S=2m_0x^2/\chi^2\bigl(1+o(1)\bigr)$. Under
these conditions $\Phi(S,x)=O(x^{\,n-1}S^{3/2})$ is negligible with respect to
$\Phi(1,x)=-t_s'(x)$ in \eqref{eq:Dexact}, and the $O(x^2)$ error in
$\chi^2=2m/r+O(x^2)$ is negligible with respect to $\chi^2$ itself. Trajectories
for which $S$ stays bounded away from zero are precisely those handled by
Lemma~\ref{lem:main}.

Note that $\Lambda\ge0$ by \eqref{eq:tsprime}, so that the dangerous term in
\eqref{eq:sys2} is the one carrying the factor $-\tfrac32\chi y$: it can make
$\mathrm du/\mathrm ds$ negative, hence $u$ unbounded backwards, only if $\chi$
and $\Lambda$ are simultaneously of order one; in the covered case Lemma~\ref{lem:main} excluded this.

Equation \eqref{eq:Lambdachi} shows that $n=3$ is distinguished: it is the only
case in which $A$ has a finite, non-zero limit, $A(x)\to3a/(2m_0)$ as $x\to0^+$,
so that \eqref{eq:sys1}--\eqref{eq:sys2} is asymptotically autonomous. The line
$y=1$, i.e.\ $u=+\infty$, is invariant, and the equilibria of the limiting system
on it are the solutions $\chi_*\in(0,1)$ of
\begin{equation}\label{eq:gdef}
  g(\chi_*)=A,\qquad g(\chi):=\frac{2+\chi}{\chi^3\,(1-\chi)} .
\end{equation}
The function $g$ is smooth and positive on $(0,1)$, diverges at both endpoints,
and has a single minimum at
\begin{equation}\label{eq:chic}
  \chi_{\mathrm c}=\sqrt3-1,\qquad
  g(\chi_{\mathrm c})=\frac{26+15\sqrt3}{2}=\frac{3\,a_{\mathrm c}}{2m_0}.
\end{equation}
This is the same function related to the marginally bound
case~\cite{GiamboLucamarini2024}: the energy function does not affect it, since
$k=O(x^2)$ is negligible in the balance leading to \eqref{eq:sys1}.
Therefore \eqref{eq:gdef} has two roots $\chi_1<\chi_{\mathrm c}<\chi_2$ when
$a>a_{\mathrm c}$, a single degenerate root $\chi_{\mathrm c}$ when
$a=a_{\mathrm c}$, and none when $a<a_{\mathrm c}$. 

The stability of these equilibria along the invariant line is controlled by the
right-hand side of \eqref{eq:sys2} evaluated at $(\chi_*,1)$, which we denote by $\Xi(\chi_*)$, namely
\begin{equation}\label{eq:Xi}
  \Xi(\chi_*):=1+\frac{(2+\chi_*)}{1-\chi_*}\left(1-\frac{3\chi_*}{2}\right)
             =\frac{3\bigl(2-2\chi_*-\chi_*^2\bigr)}{2\,(1-\chi_*)}
             =-\tfrac12\,\chi_*^4(1-\chi_*)\,g'(\chi_*) .
\end{equation}
The linearisation of \eqref{eq:sys1}--\eqref{eq:sys2} at $(\chi_*,1)$, in the
variables $(\chi,1-y^2)$, has eigenvalues $-\Xi$ and $-2\Xi$. Hence the
equilibrium is a node, and by the last equality in \eqref{eq:Xi} it is attracting
as $s\to-\infty$ if and only if $g'(\chi_*)>0$, i.e.\ if and only if
$\chi_*>\chi_{\mathrm c}$. The larger root $\chi_2$ therefore attracts an open set
of trajectories backwards, while the smaller root $\chi_1$ repels them. In the degenerate case $a=a_{\mathrm c}$, one has $\Xi(\chi_{\mathrm c})=0$ and
the equilibrium is not hyperbolic, so that linearisation is inconclusive; this
case, however, belongs to the covered regime and is already settled by
Lemma~\ref{lem:main}.

The linearisation above, however, only solves the case $n=3$, in which $A$ has a
finite limit and \eqref{eq:sys1}--\eqref{eq:sys2} is asymptotically autonomous.
For $n\in\{1,2\}$
the coefficient $A=A(x)$ varies, the solutions of \eqref{eq:gdef} are no longer
fixed points, and the larger one moves as $\chi_2(x)\to1^-$ when $x\to0^+$: no
linearisation is available. What survives in both cases is a trapping property,
which requires no autonomy, and which we therefore isolate as a statement about
the planar system alone. Recall $g$ from \eqref{eq:gdef} and
$\chi_{\mathrm c}=\sqrt3-1$.

\begin{lemma}\label{lem:trap}
Let $A_0>g(\chi_{\mathrm c})$. Then there exist $y_*\in\bigl(2/(3\chi_{\mathrm c}),1\bigr)$
and $x_\sharp\in(0,x_b]$ with the following property. Along any solution of
\eqref{eq:sys1}--\eqref{eq:sys2} whose coefficient satisfies $A(x)\ge A_0$, and
as long as $0<x\le x_\sharp$ and $(\chi,y)\in[\chi_{\mathrm c},1)\times[y_*,1]$,
\[
  \chi=\chi_{\mathrm c}\ \Longrightarrow\ \frac{\mathrm d\chi}{\mathrm ds}<0,
  \qquad\qquad
  y=y_*\ \Longrightarrow\ \frac{\mathrm du}{\mathrm ds}<0 .
\]
Consequently the trajectory cannot cross either of the two lines
$\chi=\chi_{\mathrm c}$ and $y=y_*$ along the backward evolution $s\to-\infty$.
\end{lemma}

\begin{proof}
We treat the two lines separately, since each of the two conditions is needed on
one of them only.

\emph{The line $\chi=\chi_{\mathrm c}$.} Here $y\ge y_*>2/(3\chi_{\mathrm c})>
\chi_{\mathrm c}$, so that $y-\chi>0$. Writing $\Lambda=A\chi^3$ and clearing
denominators, \eqref{eq:sys1} gives $\mathrm d\chi/\mathrm ds\le0$ if and only if
$2y+\chi\le A\chi^3(y-\chi)$, that is, if and only if $A\ge G(\chi,y)$, where
\begin{equation}\label{eq:GH}
  G(\chi,y):=\frac{2y+\chi}{\chi^3\,(y-\chi)},\qquad
  H(\chi,y):=\frac{2}{\chi^3\,(3\chi y-2)} .
\end{equation}
Since $\partial_yG=-3\chi\bigl[\chi^3(y-\chi)^2\bigr]^{-1}<0$, the function $G$
is strictly decreasing in $y$, and the requirement is most restrictive at
$y=y_*$: it suffices that $A\ge G(\chi_{\mathrm c},y_*)$.

\emph{The line $y=y_*$.} Here \eqref{eq:sys2} gives $\mathrm du/\mathrm ds\le0$
if and only if $\tfrac32A\chi^4y\ge1+A\chi^3$, that is, if and only if
$A\ge H(\chi,y)$, the quantity $3\chi y-2$ being positive because
$\chi\ge\chi_{\mathrm c}$ and $y_*>2/(3\chi_{\mathrm c})$. Moreover
$\partial_\chi\bigl[\chi^3(3\chi y-2)\bigr]=6\chi^2(2\chi y-1)>0$ on
$[\chi_{\mathrm c},1)$, again because $\chi y>\chi_{\mathrm c}y_*>2/3>1/2$, so
that $H$ is strictly decreasing in $\chi$ and the requirement is most restrictive
at $\chi=\chi_{\mathrm c}$: it suffices that $A\ge H(\chi_{\mathrm c},y_*)$.

Both conditions are therefore met as soon as
\begin{equation}\label{eq:corner}
  A_0>\max\bigl\{G(\chi_{\mathrm c},y_*),\,H(\chi_{\mathrm c},y_*)\bigr\}.
\end{equation}
Finally $G(\chi,1)=g(\chi)$ and $H(\chi,1)=h(\chi)$, where
\begin{equation}\label{eq:hdef}
  h(\chi):=\frac{2}{\chi^3(3\chi-2)},
\end{equation}
and a direct computation gives
\begin{equation}\label{eq:gminush}
  g(\chi)-h(\chi)=\frac{3\,P(\chi)}{\chi^3(1-\chi)(3\chi-2)},
  \qquad P(\chi):=\chi^2+2\chi-2 ,
\end{equation}
whose only zero in $(0,1)$ is $\chi_{\mathrm c}$. Hence
$g(\chi_{\mathrm c})=h(\chi_{\mathrm c})$, the right-hand side of
\eqref{eq:corner} tends to $g(\chi_{\mathrm c})$ as $y_*\to1^-$, and
\eqref{eq:corner} holds for $y_*$ close enough to $1$. Finally, the inequalities established above are strict, so that the right-hand
sides of \eqref{eq:sys1}--\eqref{eq:sys2}, computed without the $o(1)$ terms, are
bounded away from zero on the two lines; choosing $x_\sharp$ small enough, the
$o(1)$ terms are negligible and the conclusion is the same for the full system.
\end{proof}

\begin{proposition}\label{prop:naked}
Let Assumption~\ref{ass:data} hold and assume that the central singularity is
naked, i.e.\ $n\in\{1,2\}$, or $n=3$ with $a>a_{\mathrm c}$. Then, for $x_b$
small enough, there exists $y_*\in(0,1)$ such that, for every $y_0\in[y_*,1]$,
the solution of \eqref{eq:LTBdyn1}--\eqref{eq:LTBdyn2} with initial data
$(x_b,y_0)$ is defined on the whole maximal backward interval
$(t_{\inf},t_0]$, with $t_{\inf}=t_s(0)$, and satisfies there
\[
  y(t)>0,\qquad \chi\ge\chi_{\mathrm c}=\sqrt3-1,\qquad
  \lim_{t\to t_{\inf}^+}x(t)=0 .
\]
In particular, the conclusion of Lemma~\ref{lem:main} fails, and for $y_0<1$ the
corresponding hypersurface is a \emph{timelike} photon surface extending
backwards from $\Sigma$ up to the central singularity.
\end{proposition}

\begin{proof}
Take $x$ as independent variable and $(\chi,y)$ as dependent ones, which can be done as long as $y>0$ by \eqref{eq:ODEt}, and set
$\mathcal R:=[\chi_{\mathrm c},1)\times[y_*,1]$.

The datum lies in $\mathcal R$. By \eqref{eq:chi0} one has
$\chi(t_0)^2=\bigl(k(x_b)+\tfrac23\bigr)/\bigl(1+k(x_b)\bigr)$, which is no
longer exactly $2/3$, but tends to it as $x_b\to0^+$ for every choice of the
data; hence $\chi(t_0)\in(\chi_{\mathrm c},1)$ for $x_b$ small, while
$y_0\ge y_*$ by hypothesis.

Of the four sides of $\partial\mathcal R$, two are already excluded: $y=1$ is an
invariant line of \eqref{eq:sys1}--\eqref{eq:sys2}, and $\chi=1$ is never reached
by Lemma~\ref{lem:hor}, which holds regardless of the nature of the central
singularity and also guarantees that the system stays regular for $x>0$. The
other two are governed by Lemma~\ref{lem:trap}, provided its hypotheses hold,
which is the case for $x_b\le x_\sharp$.

On $\mathcal R$ one has
$S=r/x=2m_0x^2\chi^{-2}(1+o(1))\le2m_0x_b^2\chi_{\mathrm c}^{-2}(1+o(1))$, so
that $S$ is as small one wants for $x_b$ small and the expansions
\eqref{eq:sys1}--\eqref{eq:sys2} are valid on $\mathcal R$; and by
\eqref{eq:Lambdachi} the coefficient $A(x)$ is built on $a_n$, which through
\eqref{eq:parameters} depends on $k_0$ and $k_n$. Choosing $x_b$ small we may
therefore assume $A(x)\ge A_0>g(\chi_{\mathrm c})$ on $(0,x_b]$: for $n=3$ this
is precisely the naked condition $a>a_{\mathrm c}$, by \eqref{eq:chic}, while for
$n\in\{1,2\}$ it is automatic, since $A(x)\to+\infty$ (note that
Lemma~\ref{lem:trap} gives both cases at once and no autonomy of the system
is required).

Hence $\mathcal R$ is invariant along the backward evolution, and since it is
bounded the solution is defined for every $x\in(0,x_b]$ with $y\ge y_*>0$: in
particular $y$ never vanishes and $x(t)\to0$. 

Finally, by \eqref{eq:chidef}
one has $2m/r=\chi^2-k(1-\chi^2)$, rather than $\chi^2$, so that
$2m/r\ge\chi_{\mathrm c}^2-|k(x)|$, and therefore
\[
  r\;\le\;\frac{2m(x)}{\chi_{\mathrm c}^2-|k(x)|}\;=\;O\bigl(m(x)\bigr)
  \qquad\text{and}\qquad S=\frac rx=O(x^2)\longrightarrow0 .
\]
By
\eqref{eq:Fdef}, $t_s(x)-t=F(S,x)\to0$ as $S\to0$, whereas a trajectory reaching
the regular centre would have $S$ bounded away from zero and hence
$t_s(x)-t$ bounded away from zero as well, as in Lemma~\ref{lem:Smin}. Since
$t_s$ is continuous up to $x=0$ by \eqref{eq:tsrescaled}, we conclude that
$t\to t_s(0)$, that is the trajectory reaches the central singularity.
\end{proof}

\begin{remark}\label{rem:limitradius}
Using $\mathcal R$ we can describe
more precisely the behaviour of the  trajectories. For $n=3$ the coefficient $A$ is constant and the equilibria of
\eqref{eq:sys1} on the invariant line $y=1$ are the roots of $g(\chi)=A$, namely
$\chi_1<\chi_{\mathrm c}<\chi_2$; by \eqref{eq:Xi} the larger one is a node
attracting the backward flow, and the trajectories converge to it, so that
$y\to1$ and $r\to2m/\chi_2^{\,2}$. The areal radius thus approaches the
singularity with a fixed ratio $r/m\to2/\chi_2^{\,2}$, decreasing from
$2+\sqrt3\simeq3.73$ to $2$ as $a$ ranges over $(a_{\mathrm c},+\infty)$; the
same $\chi_2$ fixes the sharp barrier \eqref{eq:xisharp} of
Section~\ref{sec:NS}. For $n\in\{1,2\}$ one has $A(x)\to+\infty$, the larger root
moves as $\chi_2(x)\to1^-$, and the family degenerates onto the apparent horizon.

Two further comments are in order. The threshold $y_*$ in
Lemma~\ref{lem:trap} is far from optimal, as numerical integration of
\eqref{eq:sys1}--\eqref{eq:sys2} may show, but all that matters here is that this set is non-empty
and open. Also observe that the threshold does not disappear for small clouds: by
\eqref{eq:chi0} the datum in $\chi$ tends to $\sqrt{2/3}$ independently of $x_b$
and of the free functions, so that $y_0^*$ does not tend to $1$ as $x_b\to0$ and
the loss of uniqueness cannot be removed by shrinking the cloud.

Finally, it is worth noticing that in these models a single polynomial governs three apparently distinct
features of the problem. Besides \eqref{eq:gminush}, one has
\[
  g'(\chi)=\frac{3\,P(\chi)}{\chi^4(1-\chi)^2},
  \qquad
  \Xi(\chi)=-\frac{3\,P(\chi)}{2\,(1-\chi)} ,
\]
so that $P(\chi)=\chi^2+2\chi-2$ locates at once the minimum of $g$, hence the
critical value $a_{\mathrm c}$ of the causal classification; the crossing of $g$
and $h$, hence the corner of $\mathcal R$ at which both barriers of
Lemma~\ref{lem:trap} become critical; and the change of stability of the
equilibria on $y=1$. 
\end{remark}

Collecting Lemma~\ref{lem:main}, \eqref{eq:recross} and
Proposition~\ref{prop:naked}, we obtain the main result of this section.

\begin{theorem}
\label{thm:nullextension}
Let Assumption~\ref{ass:data} hold, and consider the extensions of the exterior
photon surface $r=3M$ into the interior of the cloud, i.e.\ the solutions of
\eqref{eq:LTBdyn1}--\eqref{eq:LTBdyn2} with initial data $(x_b,y_0)$,
$|y_0|\le1$, at the matching time \eqref{eq:t0def}.
\begin{enumerate}
  \item[(i)] If the central singularity is covered, the only admissible extension
  is the null hypersurface
  \begin{equation}
    \mathcal S=\bigl\{(t,x(t),\theta,\varphi):\ t\le t_0,\ (\theta,\varphi)\in S^2\bigr\},
  \end{equation}
  where, for each fixed $(\theta,\varphi)$, the curve $(x(t),1)$ solves
  \eqref{eq:LTBdyn1}--\eqref{eq:LTBdyn2} with $x(t_0)=x_b$, $y(t_0)=1$.
  \item[(ii)] If the central singularity is naked, set
  \[
    y_0^*:=\inf\bigl\{\,\bar y\in(0,1):\ \text{for all } y_0\in(\bar y,1)
    \text{ the trajectory satisfies } y(t)>0 \text{ for all } t<t_0\,\bigr\}.
  \]
  Then $y_0^*<1$ by Proposition~\ref{prop:naked}, and every $y_0\in(y_0^*,1]$
  gives an extension which reaches the central singularity. The admissible extension is therefore not unique, and it is
  timelike for $y_0<1$.
  \end{enumerate}
In particular, the uniqueness of the null extension characterises the black hole
case.
\end{theorem}

\begin{remark}
Setting $y\equiv 1$ in~\eqref{eq:LTBdyn1}, the photon surface is generated by
the outgoing radial null geodesic equation
\begin{equation}
  \frac{dt}{dx} = \frac{r_x(t(x),x)}{\sqrt{1+k(x)}}, \qquad t(x_b)=t_0,
  \label{eq:nullgeok}
\end{equation}
which extends the outgoing null geodesic equation of \textsc{Paper~I} to the case $k(x)\ne 0$. In the
covered case this ODE carries all the information on the photon surface, by
Theorem~\ref{thm:nullextension}(i); in the naked case it describes the null
member of the family, the timelike ones being discussed in
Section~\ref{subsec:family}.
\end{remark}

\section{Null geodesic structure and visibility of the singularity}
\label{sec:analysis}

By Theorem~\ref{thm:nullextension}, in the covered case the photon surface is
unique and null, so that its qualitative behaviour is entirely encoded in the
outgoing radial null geodesic ODE \eqref{eq:nullgeok}; in the naked case the null
hypersurface is still an admissible extension, and we shall study it first, the
timelike ones being discussed in Section~\ref{subsec:family}. The analysis splits
naturally according to the causal nature of the central singularity, as classified
in Section~\ref{subsec:causal}, through the two parameters $n$ and $a$ defined in
\eqref{eq:parameters}.

\begin{figure}
    \centering
    \includegraphics[width=0.8\linewidth]{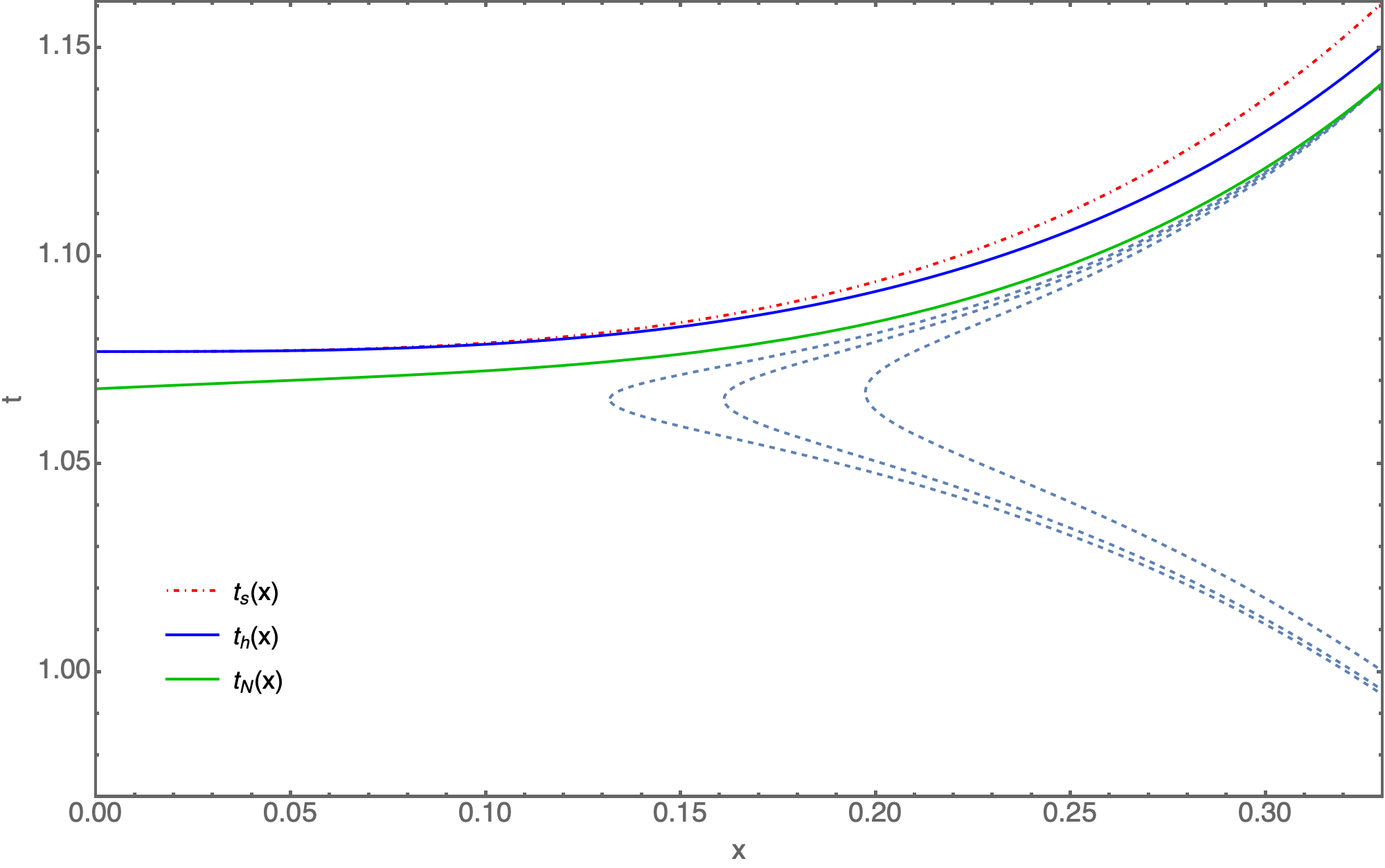}
    \caption{The null photon surface (green curve), in the case of a covered central singularity, extends back to the regular centre $t<t_s(0)$. It is the only possible extension: the dashed blue curves represent three possible extensions with $y_0<1$ (in the figure $y_0=0.9,0.93,0.95$, respectively) which eventually reach the boundary $x_b$ for a smaller comoving time. With reference to \eqref{eq:massprofile}--\eqref{eq:kprofile}, here we have $m_0=2/9,\,m_3=-1/10,\,k_0=-1/10,\,k_3=-2$. Here the parameters \eqref{eq:parameters} of Section~\ref{sec:mbvsgen} are $n=3$ and $a=2.05<a_{\mathrm c}=3.85$.}
    \label{fig:BH}
\end{figure}

\subsection{The covered (black hole) case}

Assume that the central singularity is covered, i.e.\ $n\ge4$, or $n=3$ with
$a\le a_{\mathrm c}$. By Theorem~\ref{thm:nullextension}(i) the photon surface is
the unique null extension, and the results of~\cite{Giambo2003} imply that
every outgoing radial null geodesic starting in the interior of the cloud with
$t<t_h$ extends backwards to the regular centre $x=0$ \emph{without} reaching
the central singularity (recall that we used this property
in the proof of Lemma~\ref{lem:Smin}, to bound $S$ away from zero).

The same comparison argument applies, verbatim, to the solution $t(x)$ of
\eqref{eq:nullgeok}. Since by construction $t(x_b)=t_0<t_h(x_b)$, and since
$t_h$ is a subsolution of~\eqref{eq:nullgeodesic}, we obtain
\begin{equation}
  t(x) < t_h(x) \qquad \text{for all } x\in(0,x_b],
\end{equation}
so that the photon surface stays everywhere below the apparent horizon. The
solution $t(x)$ extends backwards from $x_b$ to $x=0$, with
$\lim_{x\to 0^+}t(x)=t(0)<t_s(0)$. Therefore, in the black-hole case,
\textbf{the photon surface terminates at the regular centre} $x=0$, consistently
with the absence of any photon escaping the central singularity.

\medskip

\subsection{The naked singularity case}\label{sec:NS}

Assume now that the central singularity is naked, i.e.\ $n\in\{1,2\}$, or $n=3$
with $a>a_{\mathrm c}$. In this regime it is convenient to use the so-called \emph{area-radius coordinates} $(r,x)$, in which the
outgoing radial null geodesic equation takes the form
$\mathrm dr/\mathrm dx=F(x,r)$, see \cite{Giambo:2002xc,Giambo2003}. Because the
solution is collapsing, $r$ plays the role of a \emph{reversed} time, so that the
apparent horizon becomes a \emph{supersolution} of the null geodesic equation in
this framework; by \eqref{eq:massprofile} it behaves like
$r_h(x)=2m(x)\simeq2m_0x^3$ near the centre. Accordingly, in this regime the
central singularity is naked if and only if there exists a curve
$r_\xi(x)=\xi\,x^3$, with $\xi>2m_0$, which is a \emph{subsolution} of the null
geodesic equation in area-radius coordinates.

The constant $\xi$ is not arbitrary. Comparing with \eqref{eq:gdef}, the curve
$r=\xi x^3$ corresponds to the constant value $\chi_\xi=\sqrt{2m_0/\xi}$, and the
null geodesic flow drives $\chi$ towards $\chi_2$ as $x\to0^+$. Hence $r_\xi$ is
an admissible barrier precisely when
\begin{equation}\label{eq:xirange}
  \frac{2m_0}{\chi_2^{\,2}}\;<\;\xi\;<\;\frac{2m_0}{\chi_1^{\,2}} ,
\end{equation}
$\chi_1<\chi_2$ being the two roots of $g(\chi)=A$, which exist if and only if
$a>a_{\mathrm c}$: the naked condition is exactly the condition for
\eqref{eq:xirange} to be non-empty. The limiting, sharp value is the left
endpoint,
\begin{equation}\label{eq:xisharp}
  \xi_\ast=\frac{2m_0}{\chi_2^{\,2}}\;\in\;\bigl(2m_0,\ (2+\sqrt3)\,m_0\bigr),
\end{equation}
and the null photon surface reaches the singularity along $r\simeq\xi_\ast x^3$,
that is with $r/m\to2/\chi_2^{\,2}$.
For $n\in\{1,2\}$ one has $A(x)\to+\infty$, so that $\chi_1\to0$ and
$\chi_2\to1^-$: the interval \eqref{eq:xirange} degenerates into the condition
$\xi>2m_0$ already stated above, and it is only for $n=3$ that it is genuinely
restrictive.

Combining the subsolution $r_\xi$ with the supersolution $r_h$, the solution
$r(x)$ of the null geodesic equation is squeezed between the two in a right
neighbourhood of $x=0$:
\begin{equation}\label{eq:twosided}
  2m(x)\;\le\;r\bigl(t(x),x\bigr)\;\le\;\xi\,x^3,\qquad x\in(0,x_b].
\end{equation}
Choosing $x_b$ small enough that $t_h(x_b)<t_0<t\bigl(r_\xi(x_b),x_b\bigr)$,
which by \cite{Giambo2003} can always be arranged in the regime under
consideration, and reading the result back in comoving coordinates, the null
radial geodesic $t(x)$ extends all the way to $x=0$, with
$\lim_{x\to0^+}t(x)=t_s(0)$. This is in agreement with
Proposition~\ref{prop:naked}, of which it provides an independent derivation in
the case $y_0=1$. What the present argument adds is the estimate
\eqref{eq:twosided}: the trapping region of Section~\ref{subsec:key} bounds
$\chi$ from below only, whereas the two barriers $r_h$ and $r_\xi$ confine the
null photon surface on both sides, with the upper one arbitrarily close to the
sharp value \eqref{eq:xisharp}.

Furthermore, by the very same comparison argument, the outgoing radial null
geodesics $t_g(x)$ with initial data
$t_g(x_b)\in\bigl(t(r_\xi(x_b),x_b),\,t_0\bigr)$ also extend to $x=0$ with limit
$t_s(0)$, and since the initial datum ranges over an interval they form an
infinite family. These geodesics lie \emph{below} the photon surface, i.e.\
$t_g(x)<t(x)$ for every $x\in(0,x_b]$: physically, they are photons emanating
from the central naked singularity and escaping to the exterior of the cloud,
outside the photon surface. The photon surface, although it reaches the
singularity, is therefore \textbf{insufficient to cover it}, and the singularity
is visible to a distant observer.

\begin{figure}
    \centering
    \includegraphics[width=0.8\linewidth]{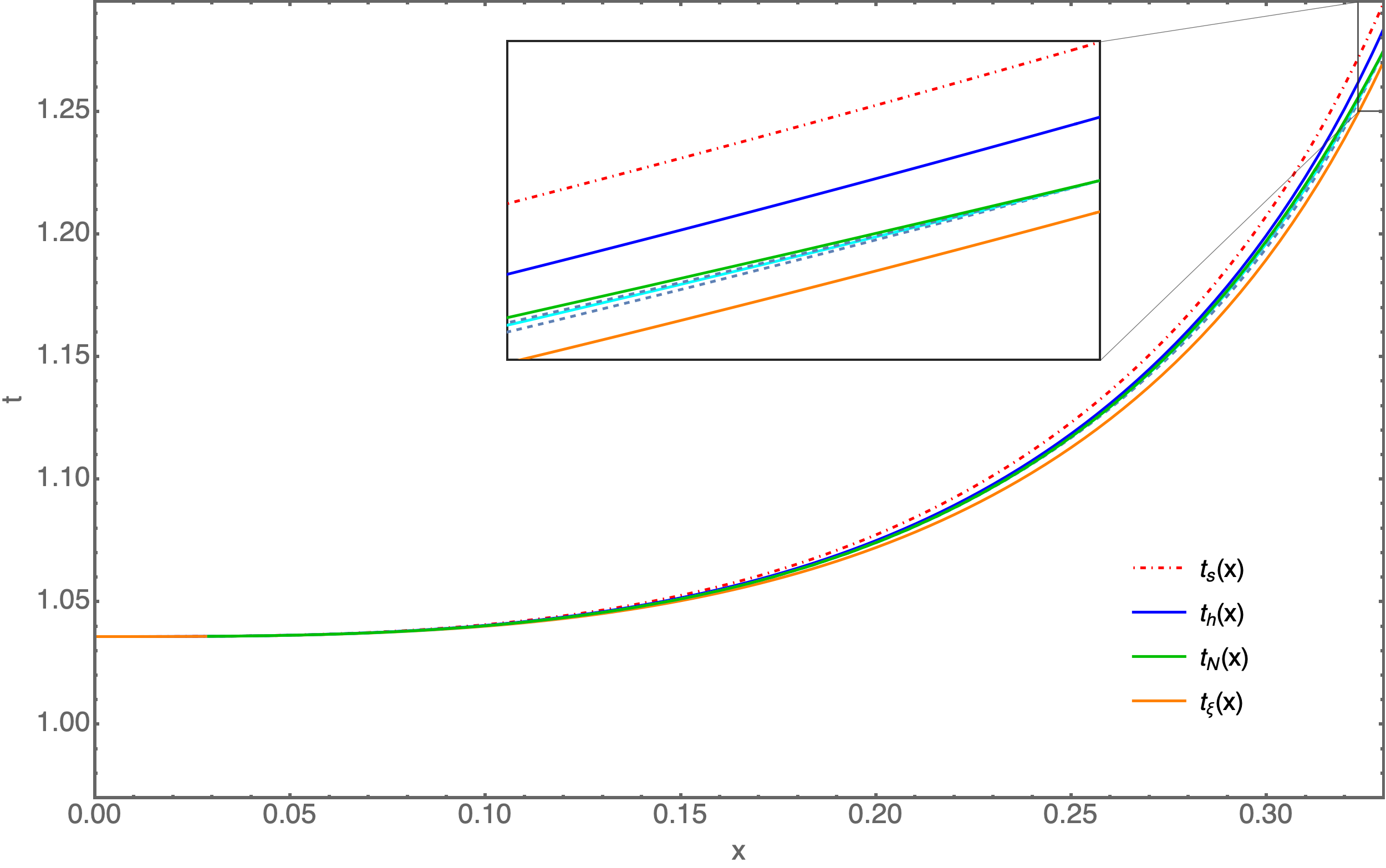}
    \caption{The null photon surface (green curve) $t_N(x)$ for a naked central singularity, bound ($k(x)<0$) case, together with three timelike extensions (two blue dashed and one solid cyan, with initial values $y_0=0.9,0.93,0.95$, respectively) extends back to the singular centre $t=t_s(0)$ because it is constrained between the apparent horizon $t_h(x)$ and $t_\xi(x)=t(\xi x^3,x)$ with a suitable choice of $\xi$ in \eqref{eq:xirange} (here $\xi=3/4$, against the sharp value $\xi_*=0.59$). With reference to \eqref{eq:massprofile}--\eqref{eq:kprofile}, here we have $m_0=2/9,\,m_3=-1/10,\,k_0=-1/20,\,k_3=-6$. Here the parameters \eqref{eq:parameters} of Section~\ref{sec:mbvsgen} are
$n=3$ and $a=4.85>a_{\mathrm c}=3.85$.}
    \label{fig:NS}
\end{figure}

\begin{theorem}
\label{thm:main}
Let Assumption~\ref{ass:data} hold, and let $t(x)$ denote the null photon surface
of Theorem~\ref{thm:nullextension}.
\begin{itemize}
  \item[(a)] If the central singularity is covered, then $t(x)<t_h(x)$ for every
  $x\in(0,x_b]$ and $\lim_{x\to0^+}t(x)<t_s(0)$: the null photon surface
  terminates at the regular centre, and no outgoing radial null geodesic issued
  from the central singularity reaches the exterior of the cloud.
    \item[(b)] If the central singularity is naked, choose $\xi$ as in
  \eqref{eq:xirange} and $x_b$ small enough that
  $t_h(x_b)<t_0<t\bigl(r_\xi(x_b),x_b\bigr)$. Then the null photon surface obeys
  the two-sided estimate \eqref{eq:twosided}, and there exists an infinite family
  of outgoing radial null geodesics $t_g(x)$ emanating from the central
  singularity, with $t_g(x)<t(x)$ for every $x\in(0,x_b]$, which escape to the
  exterior of the cloud.
  \end{itemize}
Thus the naked singularity is not merely reached by the photon surface, as
already established in Theorem~\ref{thm:nullextension}: it is \emph{visible} from 
outside the cloud, along an infinite family of null geodesics lying below the
photon surface, which the latter therefore fails to cover.
\end{theorem}

\begin{remark}\label{rem:comparisonPaperI}
Theorem~\ref{thm:main} extends to $k(x)\ne0$ the corresponding result of
\textsc{Paper~I}, where the same conclusions are established for marginally bound
collapse. That the mechanism is unchanged is to be expected, since near the
centre $k=O(x^2)$; the technique, however, is not the same. In the
marginally bound case $t(x)$ is available in closed form, and the comparison can
be carried out directly in comoving coordinates, by means of a supersolution
$t_s(0)+\xi x^n$ with $\xi$ chosen so as to lie below the apparent horizon in a
right neighbourhood of $x=0$. Here no such closed form exists, and the comparison
has to be performed in area-radius coordinates. In exchange, the admissible
barriers are no longer selected case by case: by \eqref{eq:xirange} they form an
explicit interval, determined by the roots of $g$, whose left endpoint
\eqref{eq:xisharp} is sharp and fixes the rate $r\simeq\xi_\ast x^3$ at which the
null photon surface reaches the singularity.
\end{remark}

\medskip

\subsection{The family of timelike photon surfaces}
\label{subsec:family}

By Theorem~\ref{thm:nullextension}(ii), in the naked case the null hypersurface is
not the only admissible extension. By Proposition~\ref{prop:naked} and
Remark~\ref{rem:limitradius} all members of the family share the limiting
behaviour
\begin{equation}\label{eq:familylimit}
  \frac{2m}{r}\longrightarrow\chi_2^{\,2},\qquad\text{equivalently}\qquad
  r\simeq\frac{2m_0}{\chi_2^{\,2}}\,x^3,
\end{equation}
so that they are asymptotically tangent, at the singularity, to the same curve
\eqref{eq:xisharp} as the null photon surface, and $\chi\ge\chi_{\mathrm c}$
confines the whole family to the shell
\begin{equation}\label{eq:shell}
  2\,m(x)<r<(2+\sqrt3)\,m(x)\,\bigl(1+O(x_b^2)\bigr),
\end{equation}
lying between the apparent horizon and a radius slightly larger than the photon
sphere value $r=3m(x)$, the correction being due to the energy function and absent
when $k\equiv0$.

One feature deserves emphasis: in the naked case even the null extension emanates
from the singularity, so that no regularity criterion distinguishes it from the
timelike ones.

\begin{remark}
\label{rem:unbound_examples}
The numerical examples presented in Figures~\ref{fig:BH} and~\ref{fig:NS} above
correspond to the \emph{bound} regime $k(x)<0$, i.e.\ to collapsing shells with
negative total mechanical energy. Entirely analogous examples can be constructed
in the \emph{unbound} regime $k(x)>0$, where each shell carries positive total
energy and the parametric solution is of hyperbolic type.
The qualitative behaviour of the photon surface is the same in both regimes: it
terminates at the regular centre in the covered case (Figure~\ref{fig:BH_unbound}),
and it reaches the central naked singularity in the naked case
(Figure~\ref{fig:NS_unbound}), in full agreement with Theorem~\ref{thm:main}. Also note that the examples in Figures~\ref{fig:BH_unbound} and
\ref{fig:NS_unbound} have been constructed by perturbing a model with initial
homogeneous energy density, $m(x)=m_0x^3$. For $k\equiv0$ this is the well known
Oppenheimer--Snyder collapse, for which $D\equiv0$ by \eqref{eq:Dexact}: the end
state is a black hole and the extension of the photon sphere is unique. As soon
as $k(x)/x^2$ is non-constant, however, $D\not\equiv0$ even though the mass
profile is homogeneous, the central singularity can be naked, and the whole
family of Theorem~\ref{thm:main}(b) appears. The relevant inhomogeneity is thus
that of the pair $\bigl(m(x)/x^3,\,k(x)/x^2\bigr)$, and not of the density alone:
these are configurations with no marginally bound counterpart.
\end{remark}

\begin{figure}[ht]
  \centering
  \includegraphics[width=0.8\textwidth]{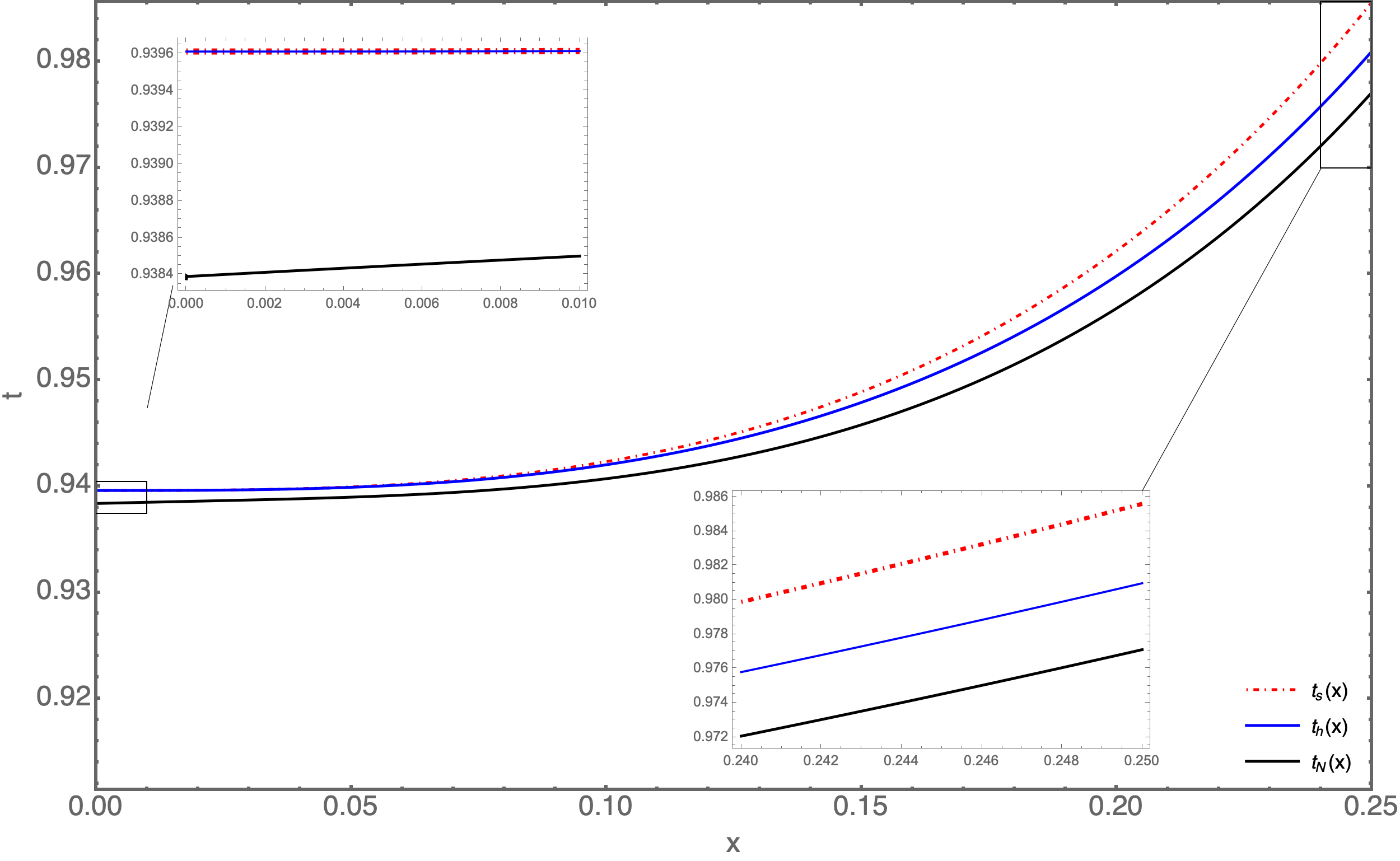}
  \caption{The photon surface $t_N$ (in black), in the case of a covered central
    singularity, \emph{unbound} ($k(x)>0$) regime. The photon surface is null and extends
    back to the regular centre without reaching the central
    singularity. With reference to~\eqref{eq:massprofile}--\eqref{eq:kprofile},
    here we take $m_0=\tfrac{2}{9}$,  $k_0=\tfrac{1}{10}$,
    $k_3=-5$, yielding parameters~\eqref{eq:parameters} $n=3$ and $a\le a_{\mathrm{c}}$. Here the boundary is set at $x_b=0.25$, so that $k(x)>0$ in $(0,x_b]$.}
  \label{fig:BH_unbound}
\end{figure}

\begin{figure}[ht]
  \centering
  \includegraphics[width=0.8\textwidth]{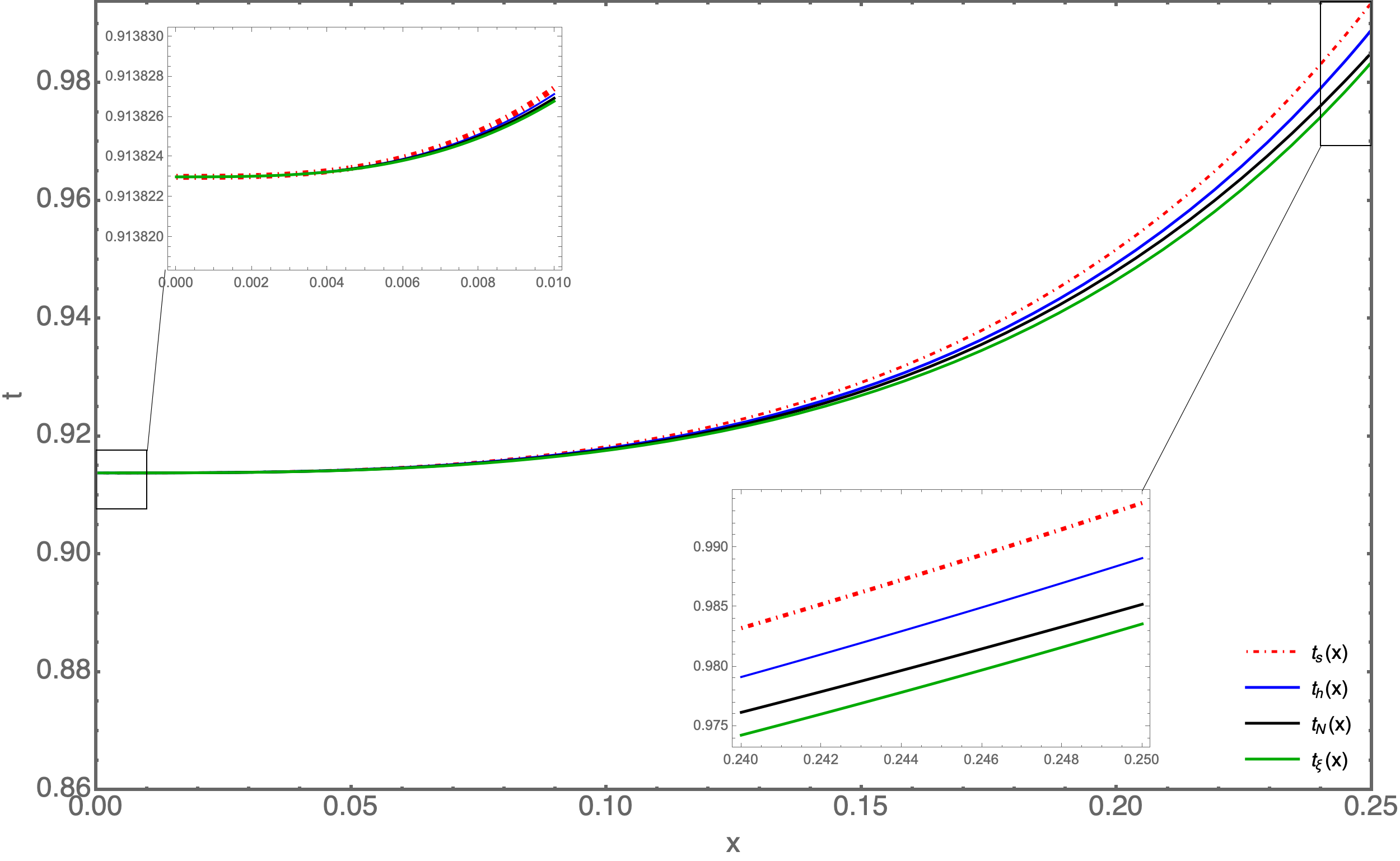}
  \caption{The null photon surface $t_N$ (in black) for a naked central singularity,
    \emph{unbound} ($k(x)>0$) regime. The null photon surface is squeezed between
    the apparent horizon $t_h(x)$ and
    the curve $t_\xi(x)=t(\xi x^3,x)$ 
    with a suitable $\xi$ in the range \eqref{eq:xirange}, here $\xi=3/4$ (against the sharp value $\xi_*=0.63$) 
    and reaches the singular centre $t\to t_s(0)$. With reference
    to~\eqref{eq:massprofile}--\eqref{eq:kprofile}, here we take
    $m_0=\tfrac{2}{9}$, $k_0=\tfrac{3}{20}$, $k_3=-9$,
    yielding parameters~\eqref{eq:parameters} $n=3$ and $a>a_{\mathrm{c}}$. The boundary is $x_b=0.25$.}
  \label{fig:NS_unbound}
\end{figure}

\section{Discussion and conclusions}
\label{sec:conclusions}

In this paper we have studied photon surfaces in general spherically symmetric
dust collapse, without assuming that the shells are marginally bound. The main
conclusions can be summarised as follows.

\begin{enumerate}
    \item Recasting the geometric condition for a timelike hypersurface to be a
  photon surface as a first-order dynamical system, and reducing that system to
  the single scalar equation \eqref{eq:ustar}, we have shown that the problem of
  extending the exterior photon sphere $r=3M$ into the cloud is governed by one
  quantity, $D=r_xr_t/r-r_{tx}$, for which \eqref{eq:Dexact} gives an exact
  expression. It vanishes identically precisely when both $m(x)/x^3$ and
  $k(x)/x^2$ are constant, so that the obstruction to uniqueness is a pure
  inhomogeneity effect of the pair of free functions, and not of the density
  alone, as we have seen for the homogeneous mass profile, that can produce a
  naked singularity and a non-unique extension as soon as the energy function is
  inhomogeneous.
    \item The extension of $r=3M$ into the interior is unique, and then null, if
  and only if the central singularity is covered
  (Theorem~\ref{thm:nullextension}); when the singularity is naked, a
  one-parameter family of genuinely timelike photon surfaces exists alongside the
  null one, all of them reaching the singularity and confined to the shell
  \eqref{eq:shell}. Uniqueness of the extension is therefore itself a
  characterisation of the black-hole case.
    \item The null photon surface reaches the central singularity if and only if the
  singularity is naked (Theorem~\ref{thm:main}). Both dichotomies are governed
  by the same threshold: the parameter $a_{\mathrm c}$ separating black hole
  from naked singularity is exactly the value at which the function
  $g(\chi)=(2+\chi)/[\chi^3(1-\chi)]$ attains its minimum, at
  $\chi_{\mathrm c}=\sqrt3-1$.   Remarkably, the initial datum of the reduced problem is universal: since
  $r(t_0,x_b)=3m(x_b)$, the variable \eqref{eq:chidef} takes at the matching
  point the value $\chi(t_0)^2=\bigl(k(x_b)+\tfrac23\bigr)/\bigl(1+k(x_b)\bigr)$,
  which tends to $\tfrac23$ as $x_b\to0$ for every choice of the data, so that the outcome depends only on $a$ and on $y_0$.
\end{enumerate}

At the technical level, the absence of a closed elementary form for
\eqref{eq:LTBevol} forces a change of strategy: the implicit solution \eqref{eq:solk},
in the rescaled form \eqref{eq:tsrescaled}, hence the estimate of Lemma~\ref{lem:Dbound}, makes transparent that the whole
$x$-dependence enters only through $m(x)/x^3$ and $k(x)/x^2$, and the additional
factor $\sqrt{1+k(x)}$ turns out to be harmless near the centre. For the null
geodesic analysis, the passage to area-radius coordinates~\cite{Giambo:2002xc,Giambo2003}
allows the comparison arguments to be carried out for the general LTB collapse,
and in fact for a wider class of collapsing solutions.

It should be stressed that what is specific to $k\ne0$ is,
first, the method -- the reduction to \eqref{eq:ustar} and the exact expression
\eqref{eq:Dexact}, which replace the explicit computations available in the
marginally bound case and single out $D$ as the one quantity responsible for the
obstruction; second, the sharper description of the null photon surface, whose
admissible barriers form the explicit interval \eqref{eq:xirange}, with a sharp
endpoint \eqref{eq:xisharp}, instead of being selected case by case; and third,
and most importantly, the fact that the loss of uniqueness can be produced by the
energy function alone, at a fixed and even homogeneous density profile, which has
no counterpart when $k\equiv0$.

From a broader perspective, the results above sharpen the link between the photon
surface and the causal nature of the central singularity. In the covered case the
photon sphere admits one and only one continuation into the collapsing region, so
that the causal boundary relevant for the early-time shadow is a single, sharply
defined null hypersurface. In the naked case this is no longer true: the
continuation is degenerate, a whole family of photon surfaces accumulates on the
singularity within the shell \eqref{eq:shell}, and outside them lies an infinite
family of outgoing null geodesics escaping to the exterior, with arbitrarily large
redshift~\cite{Dwivedi1998}. The difference between the two cases is therefore
structural rather than quantitative, and it concerns the very object that
determines the shadow boundary.

This is consistent with the expectation, supported by numerical and analytical
studies~\cite{Virbhadra:2007kw,Ortiz2015,Shaikh2019,VirbhadraEllis2002,KongEtAl2014}, that the
formation and the early-time features of the shadow of a collapsing object are
tied to the causal structure of the underlying spacetime rather than to any
specific choice of initial data or to the asymptotic regime alone. It also
suggests a concrete signature: the early-time shadow of a forming naked
singularity should not exhibit the sharp edge associated with a unique photon
surface, but rather a blurred transition region of definite width, set by
$\chi_2$ through \eqref{eq:familylimit}, leaving a potentially distinguishing
imprint for next-generation observational
instruments~\cite{EHT2019,EHT2024}.

The analysis leaves open two questions of a more technical nature. The threshold
$y_0^*$ of Theorem~\ref{thm:nullextension}(ii) has been characterised here only
as the boundary of the basin of attraction of $(\chi_2,1)$; an explicit
determination of $y_0^*$ as a function of $a$, and a proof that the
corresponding family of timelike photon surfaces foliates the shell
\eqref{eq:shell}, would complete the picture. Beyond that, an immediate natural
extension concerns the study of photon surfaces in dust collapse models with
non-vanishing tangential stresses~\cite{Magli1997}
or, more generally, with
non-zero pressure, where the structure of the central singularity and of the
apparent horizon is significantly richer,
or beyond standard general relativity,
where even the static bound $r_{\rm ph}\le3M$ from which the present analysis
starts is known not to survive unchanged~\cite{Chakraborty:2021dmu}.

\section*{Acknowledgements}

The authors thank the anonymous Referees for their careful reading of the
manuscript, which led to a substantial improvement of the results presented
here.
The authors also acknowledge the support of INdAM. In particular, RG is partially
supported by ``INdAM--GNAMPA Project'' CUP E53C25002010001. 

\bibliographystyle{plainurl}
\bibliography{ps}        


\end{document}